 \documentclass[acmtog]{acmart}
 \newcommand{\R}{\mathbb{R}}

 \newcommand{\figref}[1]{Fig.~\ref{fig:#1}}
 \newcommand{\resubb}[1]{{#1}}
 \newcommand{\resub}[1]{{#1}}
  
 \usepackage{amssymb}

\DeclareMathOperator*{\argmin}{arg\,min}

\usepackage{enumitem}
\setlist{itemsep=0pt,leftmargin=*}
\usepackage{makecell}
\usepackage{pdfrender}
\usepackage{pifont}
\usepackage{algorithm}
\usepackage[noend]{algpseudocode}
\usepackage{tikz,pgfplots}
\usepackage{csvsimple}
\usepackage{layouts}
\usepackage[super]{nth}
 \usepackage{array,multirow}
 \usepackage{wrapfig}
 \usepackage{stfloats}
 \usepackage{arydshln}

\newtheorem*{remark}{Remark}

\newcommand{\singularvalues}{\boldsymbol \sigma}

\newcommand{\cmark}{\ding{51}}%
\newcommand{\xmark}{\ding{55}}%

\AtBeginDocument{%
  \providecommand\BibTeX{{%
    \normalfont B\kern-0.5em{\scshape i\kern-0.25em b}\kern-0.8em\TeX}}}

\setcopyright{none}

\begin{document}

\title{Symmetric Volume Maps:
Order-Invariant Volumetric Mesh Correspondence with Free Boundary}
\author{S. Mazdak Abulnaga}
\email{abulnaga@mit.edu}
\affiliation{%
  \institution{Massachusetts Institute of Technology}
  \streetaddress{77 Massachusetts Ave}
  \city{Cambridge}
  \state{Massachusetts}
  \country{USA}
  \postcode{02139}
}
\author{Oded Stein}
\email{ostein@mit.edu}
\affiliation{%
  \institution{Massachusetts Institute of Technology}
  \streetaddress{77 Massachusetts Ave}
  \city{Cambridge}
  \state{Massachusetts}
  \country{USA}
  \postcode{02139}
}
\author{Polina Golland}
\email{polina@csail.mit.edu}
\affiliation{%
  \institution{Massachusetts Institute of Technology}
  \streetaddress{77 Massachusetts Ave}
  \city{Cambridge}
  \state{Massachusetts}
  \country{USA}
  \postcode{02139}
}
\author{Justin Solomon}
\email{jsolomon@mit.edu}
\affiliation{%
  \institution{Massachusetts Institute of Technology}
  \streetaddress{77 Massachusetts Ave}
  \city{Cambridge}
  \state{Massachusetts}
  \country{USA}
  \postcode{02139}
}


\begin{abstract}
Although shape correspondence is a central problem in geometry processing, most methods for this task apply only to two-dimensional surfaces.  The neglected task of \emph{volumetric} correspondence---a natural extension relevant to shapes extracted from simulation, medical imaging, and volume rendering---presents unique challenges that do not appear in the two-dimensional case. In this work, we propose a method for mapping between volumes represented as tetrahedral meshes. Our formulation minimizes a distortion energy designed to extract maps symmetrically, i.e., without dependence on the ordering of the source and target domains. We accompany our method with theoretical discussion describing the consequences of this symmetry assumption, leading us to select a symmetrized ARAP energy that favors isometric correspondences. Our final formulation optimizes for near-isometry while matching the boundary. We demonstrate our method on a diverse geometric dataset, producing low-distortion matchings that align \resub{closely} to the boundary. 


\end{abstract}




\maketitle

\section{Introduction}

  \begin{figure}
    \includegraphics[width=\linewidth]{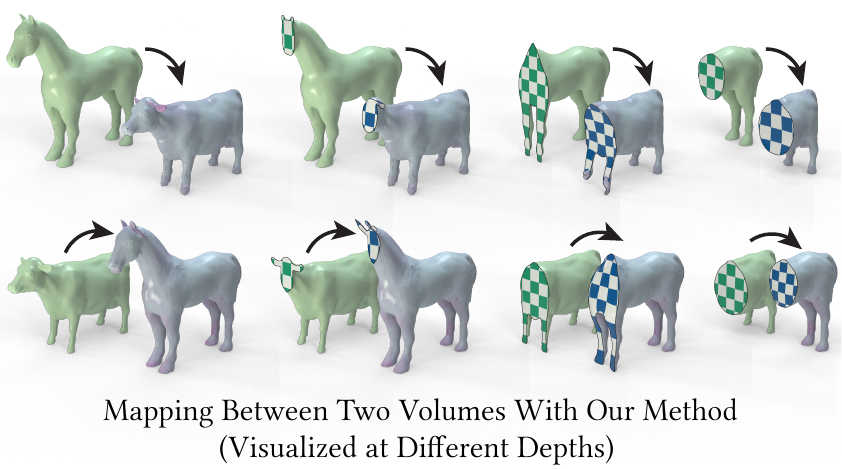}
    \caption{Our method produces low-distortion correspondences between volumes, visualized as checkerboard textures through the sliced volumes.
    \label{fig:teaser}\vspace{-.2in}
    }
  \end{figure}

Shape correspondences are at the core of many applications in graphics and geometry processing, including texture and segmentation transfer, animation, and statistical shape analysis. The central objective of these applications is to compute a dense map between two input shapes, facilitating semantically-meaningful information transfer with minimal distortion. 

The vast majority of shape correspondence algorithms focus on mapping two-dimensional surfaces. These approaches leverage geometric properties that are unique to surfaces. For example, key shape properties like curvature are defined over the entire surface domain, rather than only on the boundary as in the volumetric case. As a result, one can even find reasonable correspondences by matching geometric features directly, without incorporating any notion of distortion~\cite{ovsjanikov2010one}. Other methods use Tutte's embedding or notions of discrete conformality specific to surfaces to achieve key properties like invertibility~\cite{lipman2009mobius,schmidt2019distortion}. 

In contrast, here we consider the problem of mapping \emph{volumes} to \emph{volumes} rather than surfaces to surfaces. Volumetric correspondence is beneficial for several tasks. In graphics and CAD, boundary representations of shapes are used to represent objects, so even the input geometry used to evaluate surface-to-surface mapping techniques typically expresses a volumetric domain. Hence, finding volumetric correspondences may improve correspondences of these boundary representations, since volumetric reasoning is needed
\resub{to preserve thin features and prevent volumetric collapse; for example to prevent the candy wrapper artifact, where regions twist about a point and change orientation. In these cases, surface area is roughly maintained while volume degenerates. See Fig.~\ref{fig:didactic} (top) for an illustration using the surface mapping approach of~\citet{ezuz2019reversible}}. \resubb{ From a surface isometry perspective, the candy wrapper artifact has little distortion as only few edges have deformed. However, from a volumetric perspective, the shape's volume has completely degenerated.} \resubb{In other applications, such as medical imaging, data is acquired in a regular 3D grid and shape correspondence is used for volumetric texture transfer or alignment. Consequently, extending surface correspondences to the interior of volumetric shapes is nontrivial, so volumetric mapping approaches are needed.  }

Volumes do not share many of the geometric properties that have enabled mapping techniques for surfaces, so new approaches are needed. 
The closest existing methods to volumetric mapping tackle volumetric deformation and parameterization. In these applications, one starts with a volume in its rest pose and deforms the volume to a target domain or to conform to a set of target handle positions in a fashion that minimizes distortion.  
%
These approaches differ from volumetric mapping in several  ways. 
First, volumetric deformation and parameterization methods typically assume a reasonable initial guess (e.g., the source shape) and flexibility in the target domain (e.g., unconstrained geometry away from the handles) or specialize to a single target (e.g., a ball).  
\resubb{In contrast, in mapping, the source and target domains are geometrically distinct shapes so a reasonable initialization is not given. One may need to start with a coarse map to a known set of landmarks~\cite{aigerman2014lifted,ezuz2019reversible}}.
%
Furthermore, mapping problems are typically symmetric, in the sense
that the computed map should be invariant to the ordering of the source and target domains; there is no notion of a ``rest pose'' typical in deformation. Consequently, we seek a distortion energy that is symmetric with respect to the source and target. 

We propose an algorithm for mapping between volumes represented as tetrahedral meshes. Our method draws insight from 2D surface mapping and 3D deformation. 
%
It builds on the discretization of maps used in a state-of-the-art surface mapping algorithm~\cite{ezuz2019reversible} but requires new objective functions and optimization methods to be effective. 
In particular, we propose a set of \textit{symmetrized} distortion energies that are invariant to the domain over which the map is applied and aim to produce inversion-free, low-distortion matchings that conform closely to the boundary
(\figref{teaser}).

\begin{figure}
    \includegraphics[width=0.9\linewidth]{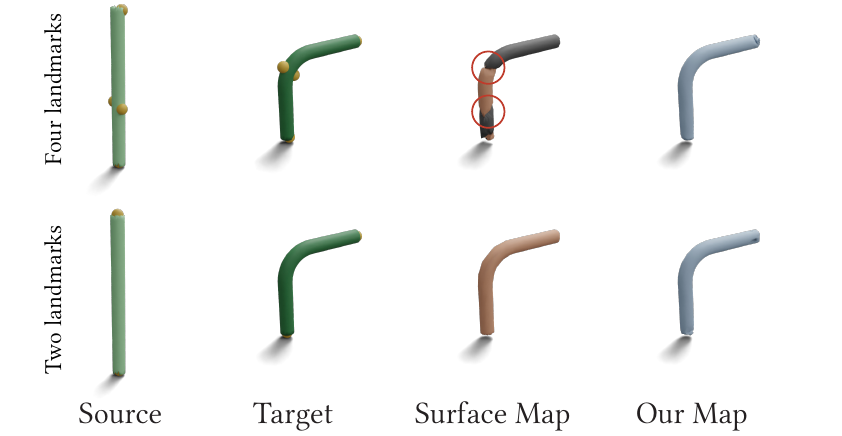}
    \caption{\resubb{Illustration of possible map degeneration when using a surface-mapping approach. Top row: Mapping using the surface-based approach of~\citet{ezuz2019reversible} initialized with four landmark points (yellow spheres) leads to the candy wrapper artifact, where regions of the mapped shape twist $180^\circ$, causing a change in orientation accompanied by a collapse in volume (red circles). The dark gray regions of the surface map show the backs of the triangles. Bottom row: mapping with two landmarks at the ends of the rods corrects the issue. In both cases, our volumetric approach maintains volumetric integrity and preserves orientation.}
    \label{fig:didactic}\vspace{-.1in}
    }
  \end{figure}
\paragraph*{Contributions.}  This paper contributes the following:
\begin{itemize}
\item We present a method for computing volumetric correspondences between far-from-isometric shapes by minimizing a symmetric distortion energy.
\item We analyze the concept of a \emph{symmetric} distortion energy, which is agnostic to the ordering of source and target domains, and provide a recipe for \textit{symmetrizing} a distortion energy. We propose a set of desirable properties for a symmetric distortion energy and analyze well-known measures of distortion within our framework.
\item We demonstrate our method on a diverse dataset of examples, showing that our method reliably extracts correspondences with low distortion.
\end{itemize}
\resub{
\subsection{Approach}

We find a dense correspondence between two volumetric shapes $M_1$ and $M_2$ represented as tetrahedral meshes. Our algorithm simultaneously optimizes for a map $\phi:M_1\to M_2$ and its (approximate) inverse $\psi\approx\phi^{-1}:M_2\to M_1$, which both take vertices of one mesh to (interiors or boundaries of) tetrahedra in the other. Our approach handles meshes of differing connectivity and facilitates finding maps between far-from-isometric shapes. 

Existing volumetric mapping methods use deformation techniques to place or repair interior tetrahedra, given a fixed map between the boundaries $\partial M_1$ and $\partial M_2$. In contrast, 
we include the boundary map as a variable. Our method can repair poorly-initialized surface maps and compute maps using only landmark correspondences as initialization.


\resubb{Our formulation is \textit{symmetric} in that the computed map is invariant to the labeling of the ``source'' and the ``target'' among $M_1$ and $M_2$. The motivation for symmetry comes from several applications where the selection of a source or target shape is unnecessary. For example, in medical imaging, one is interested in finding correspondences between brain shapes extracted from magnetic resonance images (MRI) to perform comparisons of local cortical (brain tissue) thickness~\cite{aganj2015avoiding}. Similar symmetry arises when seeking a correspondence between two humans standing in the same pose, and in general for applications seeking to align two shapes. 
The arbitrary choice of the source shape is a consequence of algorithm design rather than application need. Consequently, this choice can influence the correspondence result, introducing bias.
As shown in Fig.~\ref{fig:asymmetry}, an asymmetric method like~\cite{kovalsky2015large}  may result in unequal performance dictated by the choice of map direction. Further, the asymmetry of previous approaches in medical imaging have introduced bias in estimating the effects of Alzheimer's disease~\cite{fox2011algorithms,yushkevich2010bias,hua2011accurate}. }

\resubb{A reasonable expectation is to produce the same map--up to inversion--regardless of the choice of the source and target shape, i.e., the ordering of $M_1$ and $M_2$. One way to achieve this is to use a symmetric energy.} An energy $E$ is symmetric if $E(\phi)=E(\phi^{-1})$~\cite{cachier2000symmetrization,schmidt2019distortion}. Since $\phi^{-1}$ is challenging to compute in practice, and does not exist for maps initialized with flipped tetrahedra, we introduce $\psi \approx \phi^{-1}$ and propose a symmetric approach by optimizing $E(\phi)+E(\psi)$. Optimizing with this pair of maps is a common way of guaranteeing symmetry~\cite{christensen2001consistent,cachier2000symmetrization,ezuz2019reversible,schmidt2019distortion,schreiner2004inter}, and we show via change-of-variables that optimizing this sum is \emph{equivalent} to optimizing a different distortion energy $E^{\mathrm{Sym}}(\phi)$ on just the forward map $\phi$. 

Key to computing a high-quality map is the proper choice of distortion energy $E$ or its symmetrized counterpart $E^{\mathrm{Sym}}$. We analyze the effect of symmetrizing several widely-used distortion energies, showing that several symmetrized energies violate typical desiderata used to design mapping algorithms. For example, several symmetrized energies no longer favor local isometry. Following this analysis, we select the symmetrized ARAP energy as our distortion measure, eliminating solutions that locally favor collapsing or shrinking maps.

\begin{figure}
    \includegraphics[width=\linewidth]{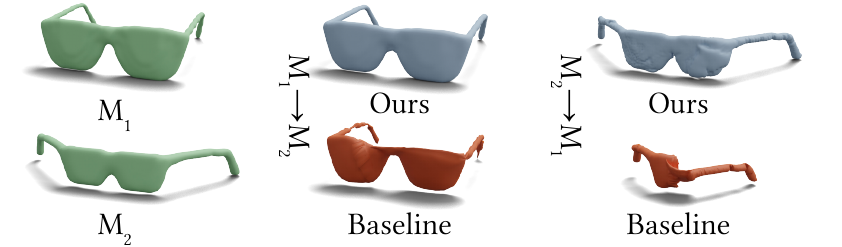}
    \caption{\resub{Comparison between our symmetric approach and an asymmetric baseline. A symmetric approach is necessary when there is no clear source or target shape to produce high-quality bidirectional maps.\vspace{-.1in}}}
    \label{fig:asymmetry}
\end{figure}

}
\section{Related Work}
Volumetric correspondence poses a new set of challenges that has not been addressed in surface-based methods. Although relatively few works consider precisely the problem tackled in this paper, we draw insights from volumetric parameterization, volumetric deformation, and surface mapping and focus our review on relevant work on these topics. 

\paragraph*{Volumetric parameterization and deformation} Parameterization and deformation algorithms provide means of deforming tetrahedral meshes into prescribed poses or domains with minimal distortion. 

A \emph{parameterization} is a deformation of a volume to a simpler domain, such as a topological ball~\cite{paille2012conformal,wang2003volumetric,yueh2019novel,abulnaga2021volumetric,garanzha2021foldover} or a polycube~\cite{fu2016inversion,paille2012conformal,xia2010direct,wang2008modeling,aigerman2013injective,li2021interactive}.  The better-studied instance of parameterization in graphics maps\resub{, possibly with cuts,}  two-dimensional surfaces (rather than volumes) into the plane; see \cite{floater2005surface,sheffer2007mesh,fu2021inversion} for discussion of this broad area of research.

In \emph{deformation}, the task is to deform a volume by moving a set of handles to a set of target positions. These methods are often based on physical models of strain~\cite{irving2004invertible} and aim to produce elastic deformations minimizing a prescribed energy choice~\cite{muller2002stable,chao2010simple,irving2004invertible,kovalsky2014controlling,fu2015amips,smith2018stable,smith2019analytic,sahilliouglu2015skuller}. \resub{In the 2D case, both skeleton-based~\cite{lewis2000pose} and physical models~\cite{nealen2006physically} can be used.} See~\cite{sieger2015shape,gain2008survey,selim2016mesh} for general discussion. 

In both problems above, one computes a deformation from the rest pose to the target. Optimization methods are used to match the target while minimizing distortion, where the distortion is measured using an energy that quantifies the deformation of the Jacobian matrix of each tetrahedron. Since these models start with a good initialization, namely the rest pose, one can optimize using a combination of energies with flip-free barriers and a constrained line search, arriving at solutions that are both flip-free and have low distortion; see e.g.~\cite{smith2015bijective} for a representative example. In contrast to these past works, we produce maps between far-from-isometric domains without an obvious effective initialization. Consequently, our choice of energies is designed to be resilient to poor initial maps that are not foldover-free.

\paragraph*{Volumetric mappings} Some methods consider the task of computing correspondences between volumetric shapes. To our knowledge, all past methods can be understood as special cases of the deformation methods where the task is to extend a fixed boundary map to the interior of a volume.

\resub{\citet{kovalsky2015large} present a local-global alternating algorithm targeting maps with bounded distortion. Their method takes an initial surface map and computes a similar map with bounded condition number. 
They demonstrate their algorithm on two volumetric correspondence examples and show one example (their Figure 11) where relaxing prescribed boundary constraints at the end of the optimization procedure can help recover from minor artifacts. 
\citet{su2019practical} also target computation of foldover-free volumetric maps with prescribed boundary; they extend the method of~\citet{kovalsky2015large} by automatically finding a suitable bound on the condition number.  Their method has impressive levels of efficiency but targets a specific notion of conformal distortion.}
\citet{stein2021splitting} propose an operator splitting technique to optimize nonconvex distortion energies to yield a flip-free parameterization; they demonstrate a few examples of volumetric correspondence. 

The approaches above require a prescribed boundary map and minimize distortion of the interior. \resubb{In contrast, our method optimizes the boundary map to minimize global distortion and does not need a bijective, orientation-preserving boundary map as an initializer.} Indeed, it is not always obvious how to design a boundary map so that the induced volumetric correspondence has low isometric disortion.  We also optimize an alternative objective function that targets symmetry and isometry rather than bounded distortion or conformal structure preservation.

A few mapping methods reduce a mapping problem between volumetric domains to a sequence of surface-mapping problems between leaves of foliations of the two domains. \citet{campen2016bijective} propose a volumetric parametrization approach relying on a foliation. Their algorithm requires the domain to be a topological ball whose tetrahedral mesh is \emph{bishellable}. \citet{cohen2019generalized} describe an alternative method to compute foliations of more-general volumetric domains using a Hele-Shaw flow along a potential function from a M\"obius inversion of the domain boundary to a sphere.
Unlike these methods that decompose the domain into surfaces, our method does end-to-end optimization of a mapping over an entire volume at once.

\paragraph*{Symmetric maps.} \emph{Symmetric} mapping methods are invariant to the ordering of the source and target shapes. Several works in 2D surface mapping do so by optimizing for the average of the forward and reverse map distortion~\cite{hass2017comparing,schreiner2004inter,ezuz2019reversible,schmidt2019distortion}. \resubb{In medical imaging, mapping is referred to as registration, where the problem is to learn a displacement field defined on a 3D grid. Symmetry, or ``inverse-consistency"~\cite{christensen2001consistent} is achieved using a similar approach of averaging the map distortions~\cite{cachier2000symmetrization,aganj2015avoiding,sabuncu2009asymmetric,leow2005inverse}, or by optimizing in a mid-space between the two images~\cite{avants2008symmetric,joshi2004unbiased}. Many of these works demonstrate that symmetry improves consistency of mapping, improves accuracy, and eliminate bias.}

We use a similar formulation to achieve symmetry. We analyze several common distortion energies symmetrized in this way and show that---surprisingly---the choice of energy can have counterintuitive consequences. In particular, distortion energies that favor isometry in one direction may not do so when optimizing their symmetrized counterparts. \resub{To prevent this undesired behavior,~\citet{hass2017comparing} developed a symmetric distortion energy that measures the distance of a conformal map from an isometry. Their distortion energy is restricted to conformal maps between genus-0 surfaces. Extending it to the volumetric case is nontrivial due to the lack of conformal maps in 3D.} 

We develop the concept of a \emph{symmetric} energy that is invariant to the choice of optimization domain over which it is taken, in the sense that the energy of the inverse map matches that of the forward map. Although it is a sensible choice in our setting, we note the term ``symmetric'' is somewhat overloaded in the parameterization and mapping literature. 
Several distortion measures have been deemed symmetric because they equally penalize scaling and shrinking, such as the symmetric Dirichlet energy~\cite{schreiner2004inter,smith2015bijective} and the symmetric ARAP energy~\cite{shtengel2017geometric}. Our analysis shows that in fact these energies do not necessarily satisfy our notion of symmetry.

\paragraph*{Surface maps} Two-dimensional surface mapping can generally be divided into (at least) three sets of approaches: methods that use an intermediate domain, methods that rely on descriptors, and methods that directly extract a map from one mesh into another. We refer the reader to one of several surveys for a broad overview~\cite{van2011survey,li2014computing,sahilliouglu2020recent}. 

The first two groups of approaches cannot be directly extended to the volumetric case. In particular, while Tutte's parameterization provides a natural means of mapping surfaces bijectively to an intermediate domain and thus provides a natural means of initializing maps in the first category, no such canonical parameterization exists for volumes.  Moreover, volumetric geometry descriptors do not appear to be sufficiently reliable for correspondence tasks.

Methods that find correspondences through an intermediate domain employ a bijective parameterization of each input to a simple domain such as the plane~\cite{kraevoy2004cross}, the sphere~\cite{gotsman2003fundamentals,haker2000conformal,lee2019dense}, or a quotient manifold~\cite{aigerman2014lifted,aigerman2015orbifold,aigerman2015seamless,aigerman2016hyperbolic,bright2017harmonic,schmidt2019distortion}. 
We also note methods like \cite{kim2011blended,lipman2009mobius}, which average multiple maps computed in a similar fashion. 
These approaches admit no obvious extension to volumes. First, the existence of a bijection to a simpler intermediate domain does not always exist. Second, many of these methods require introducing cutting seams~\cite{aigerman2015seamless}, which becomes substantially more difficult in three dimensions. Furthermore, these may not result in low-distortion maps, as minimizing the composition of the maps in the intermediate domain may result in high distortion in the final surface-to-surface map.

The second set of methods computes maps that match descriptors, possibly with added regularization. Descriptors are often distance-based~\cite{bronstein2008analysis,huang2008non}, spectral~\cite{jain2007non,mateus2008articulated,vestner2017efficient,ovsjanikov2010one}, extrinsic~\cite{ankerst19993d,salti2014shot}, or a combination~\cite{dubrovina2011approximately,kim2011blended,litman2013learning}. Many correspondence methods in this category are built on the functional maps framework~\cite{ovsjanikov2012functional,ovsjanikov2016computing}, which finds correspondences by matching functions defined on the shapes. Relatively few descriptors are available for volumetric geometry, whose structure is still inherited from the boundary surface. 

The third class of approaches directly optimize for inter-surface maps. These methods compute a map between surfaces by matching features or landmarks while minimizing distortion~\cite{schreiner2004inter,solomon2012soft,solomon2016entropic,ezuz2019reversible,mandad2017variance}.

\citet{ezuz2019reversible} produce a map between surfaces by minimizing the geodesic Dirichlet energy of the forward and reverse map and encouraging bijectivity through a reversibility energy. Our algorithm extends many of their ideas to the volumetric case. In our case, however, a new algorithm is required.


\paragraph*{Medical image registration} 
Medical image registration is a form of volumetric shape correspondence in Euclidean space. \resubb{Here, the task is
to find correspondences between two volumes defined on a dense 3D grid. The correspondence is driven by matching voxel signal intensities, for example using mutual information~\cite{klein2007evaluation} or cross-correlation~\cite{avants2008symmetric}. The optimization seeks to find a displacement field defined at the grid coordinates.} Similar to our formulation, the transformation is governed by any of several regularization terms, for example to compute a diffeomorphic transformation~\cite{beg2005computing}. We refer the reader to surveys~\cite{oliveira2014medical,sotiras2013deformablesurvey,viergever2016survey}. While both our approach and registration methods aim to find volumetric correspondences, the techniques used in medical image registration are not applicable, as they operate \resubb{on a dense Euclidean grid and are driven by intensity rather than geometry. }


\section{Mapping Problem}
\label{s:theory}
We develop a volumetric mapping method that is symmetric, in that the resulting maps are invariant to the ordering of the source and target shapes. \resub{We compute the map by minimizing an objective function that measures distortion symmetrically while satisfying a set of constraints.} In this section, we investigate the consequences of the symmetry assumption on our algorithmic design. 



\subsection{Preliminaries}

Given two bounded volumes $M_1,M_2 \subset \R^3$ with smooth boundaries $\partial M_1, \partial M_2$, we seek a map $\phi:M_1 \rightarrow M_2$. Several considerations inform our choice of $\phi$, detailed below.
Note that this problem is not the same as \emph{deformation} (sometimes referred to as ``mapping'' in past literature), which aims to find a low-distortion deformation of $M_1 \subset \R^3$ given prescribed target positions for a few handles rather than the geometry of $M_2$.

\resubb{Many algorithms for mapping and deformation can be viewed as optimizing a distortion energy of the form }
 \begin{equation}
 \label{eqn:map-distortion}
     E_f[\phi]:=\int_{M_1} f(J_{\phi}(\mathbf{x}) )\,dV(\mathbf{x}),
 \end{equation}
 where $J_{\phi}\in\R^{3\times3}$ is the map Jacobian and $dV(\mathbf{x})$ is the volume form on $M_1$.
 
 The distortion function $f:\R^{3\times 3}\rightarrow \R_{\geq 0}$ usually measures local deviation of the map from isometry. 
 Typical choices favor rigidity~\cite{rabinovich2016scalable}. For example, the as-rigid-as-possible distortion function (ARAP)~\cite{liu2008local} measures the deviation of the Jacobian from the set of rotation matrices $\mathrm{SO}(3)$: $$f_{\mathrm{ARAP}}(J) = \min_{ R\in \mathrm{SO}(3)}\|J-R\|_F^2.$$ In contrast, the Dirichlet energy functional $$f_{\mathrm{D}}(J) = \| J\|_F^2$$ favors the as-constant-as-possible map~\cite{schreiner2004inter}. Selection of the distortion function is application-dependent. For example, one might choose $f$ to model physical strain for deformation. Alternatively, one might select $f$ to encourage injectivity. 
 
 In almost all applications, $f$ is chosen to be \emph{rotation invariant}, reflecting the fact that rigid motions of $M_1$ and $M_2$ should not affect the computed map.  In this case, $f(J)$ is a function of the singular values $\singularvalues(J)$, the elements of the diagonal matrix $\Sigma$ in the singular value decomposition (SVD) $J=U\Sigma V^\top$.  In a slight abuse of notation, in our subsequent discussion we will use $f$ to denote both a function on matrices in $\R^{3\times3}$ and vectors of singular values in $\R^3$, with $f(J):=f(\singularvalues(J))$.
 
 In addition to finding a map with low distortion, we are concerned with finding one that satisfies a desired set of constraints. \resub{For example, we can constrain the boundary of the source volume to be mapped to the boundary of the target, i.e.\ $\phi(\partial M_1) = \partial M_2$}. 
 We use $\mathcal{P}$ to denote the constrained \emph{feasible set}. One might imagine other constraints, for example ensuring a set of landmark points are mapped to the pre-specified locations, further restricting $\mathcal P$. Moreover, regularizing objective terms, $\mathrm{Reg}[\phi]$ could be added. So, our optimization problem becomes

 \begin{equation}
\label{eq:opt-problem}
\begin{aligned}
\argmin_{\phi} \quad &      \int_{M_1} f(J_{\phi}(\mathbf{x}) )\, dV(\mathbf{x}) + \mathrm{Reg}[\phi]\\
\textrm{subject to} \quad & \phi \in \mathcal{P}.\\
\end{aligned}
\end{equation}

  


\begin{figure*}
    \centering
    \begin{tabular}{r@{\ }c@{\ }c@{\ }c@{\ }c@{\ }c@{\ }c}
         & Dirichlet & Symmetric Dirichlet & MIPS & Symmetric gradient & Hencky strain & ARAP \\
         \raisebox{.2in}{$f$} & 
         \includegraphics[height=.15\linewidth]{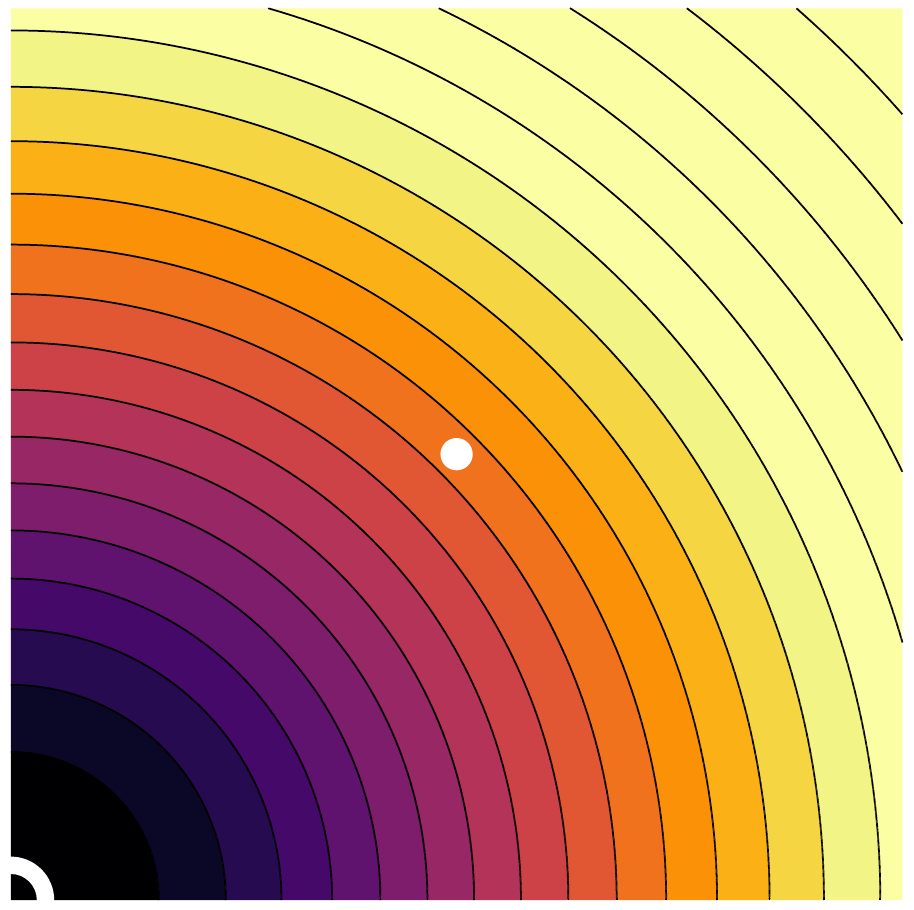}&
         \includegraphics[height=.15\linewidth]{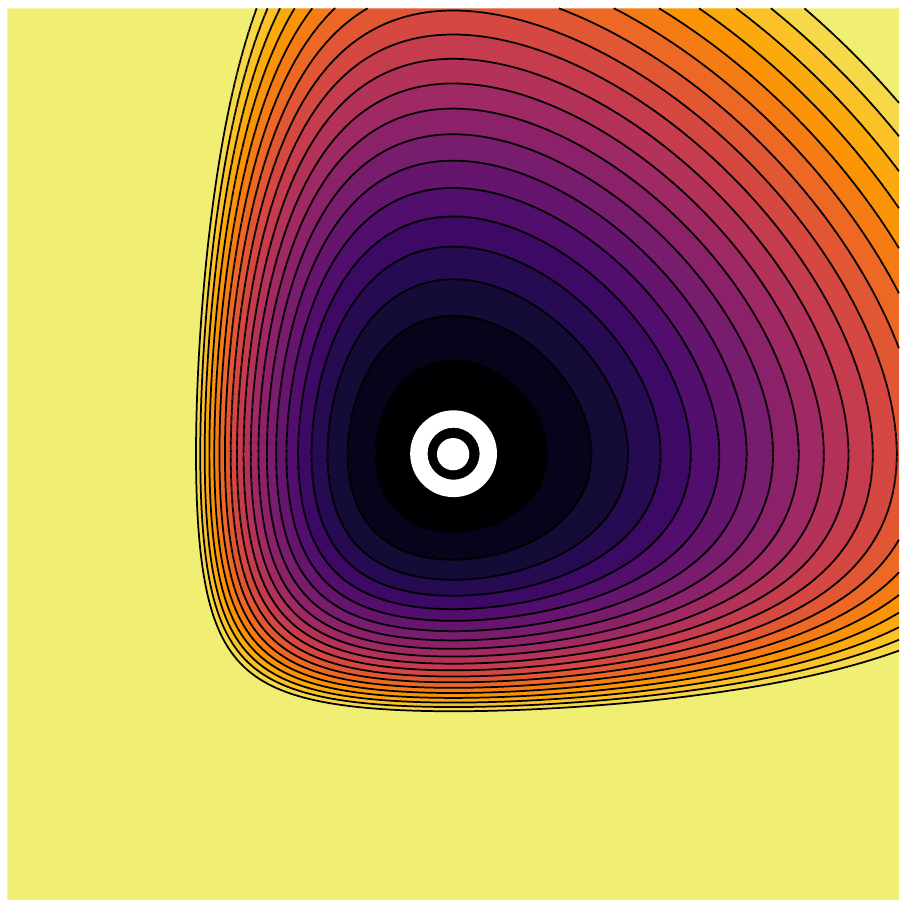}&
         \includegraphics[height=.15\linewidth]{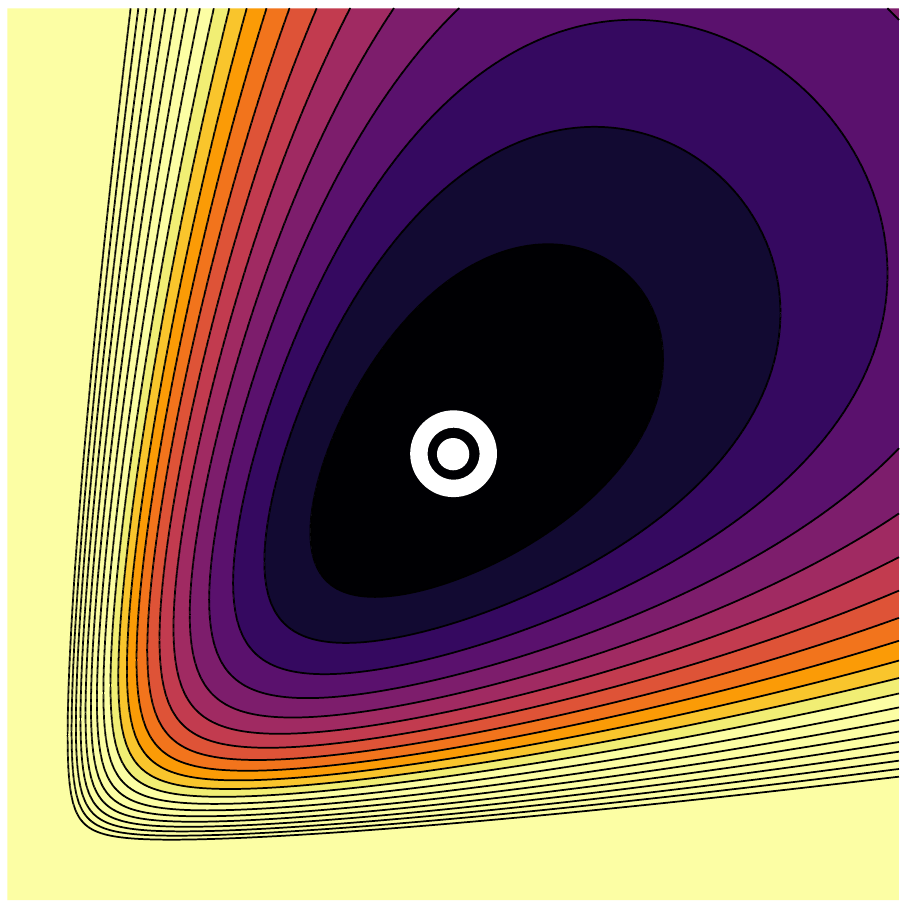}&
         \includegraphics[height=.15\linewidth]{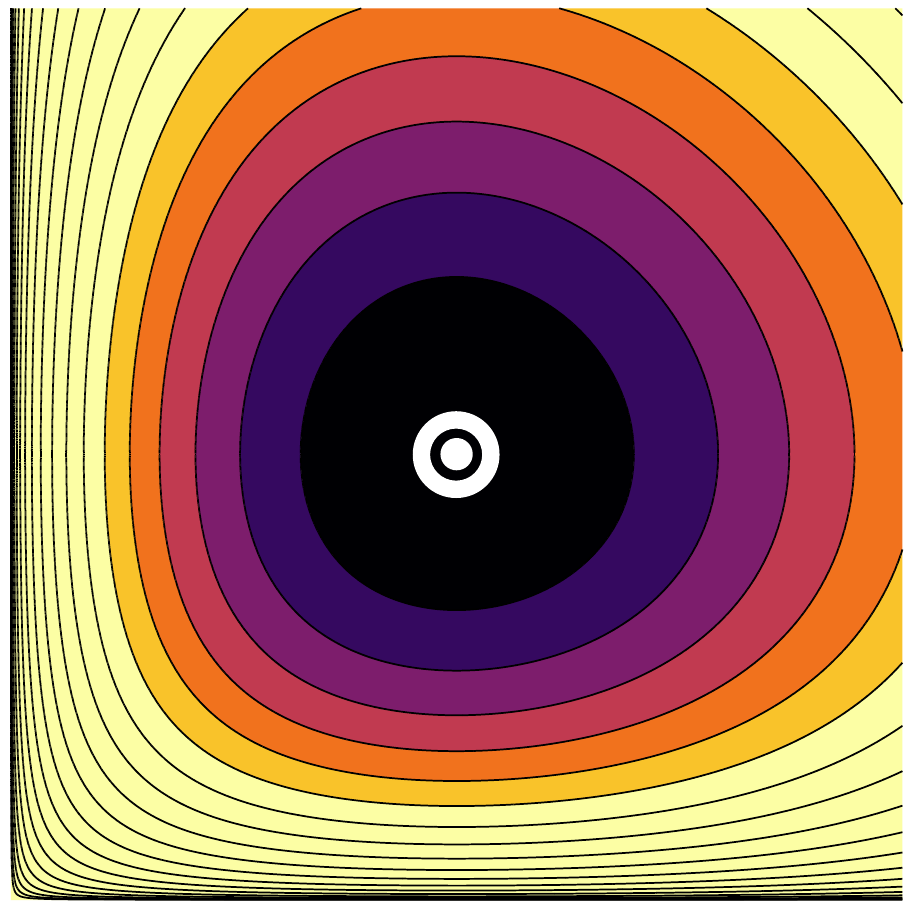}&
         \includegraphics[height=.15\linewidth]{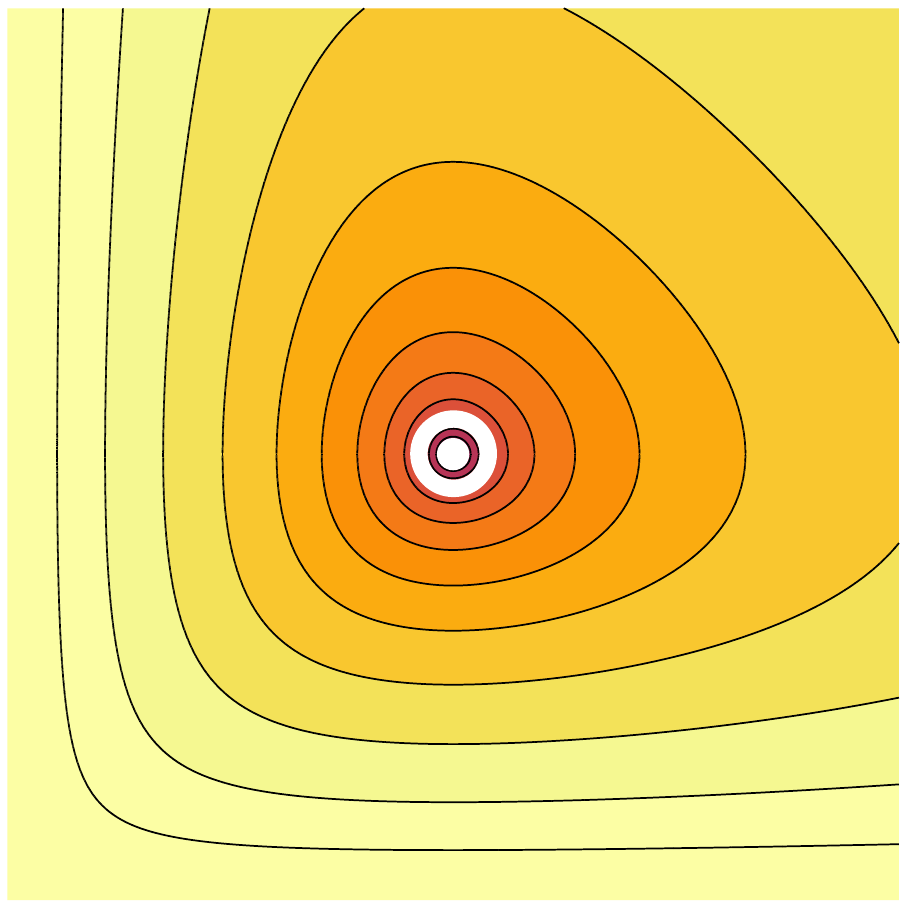}&
         \includegraphics[height=.15\linewidth]{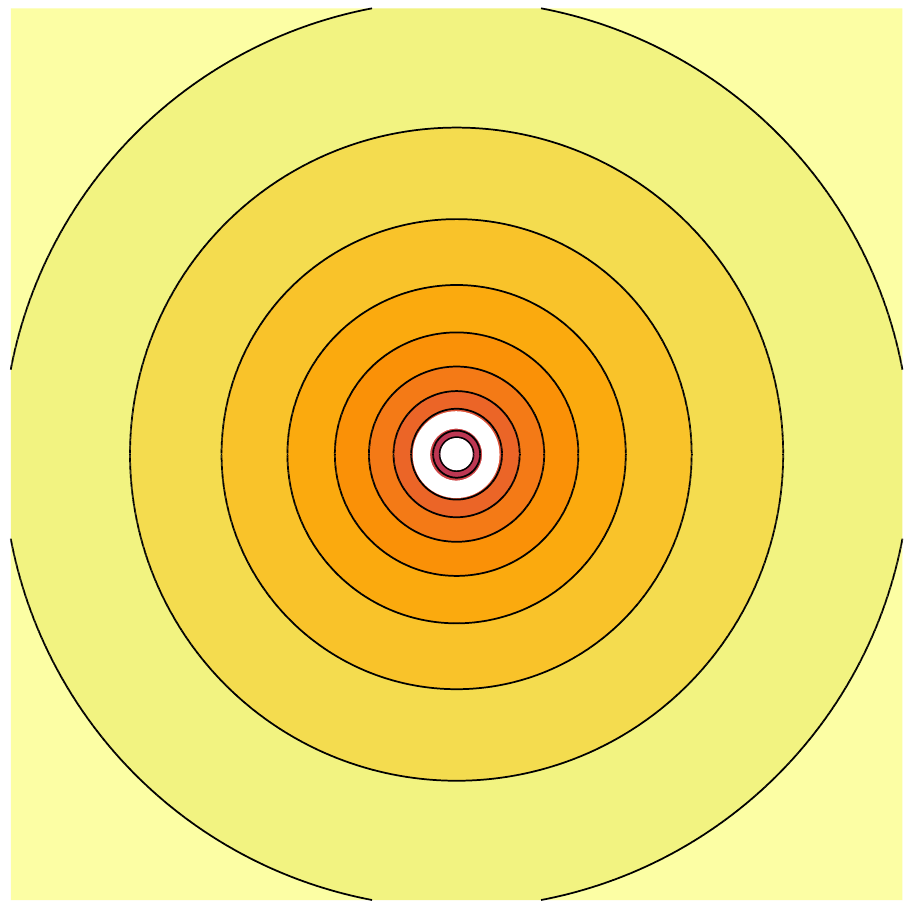}
         \\
         \raisebox{.2in}{$f^{\mathrm{Sym}}$} &
         \includegraphics[height=.15\linewidth]{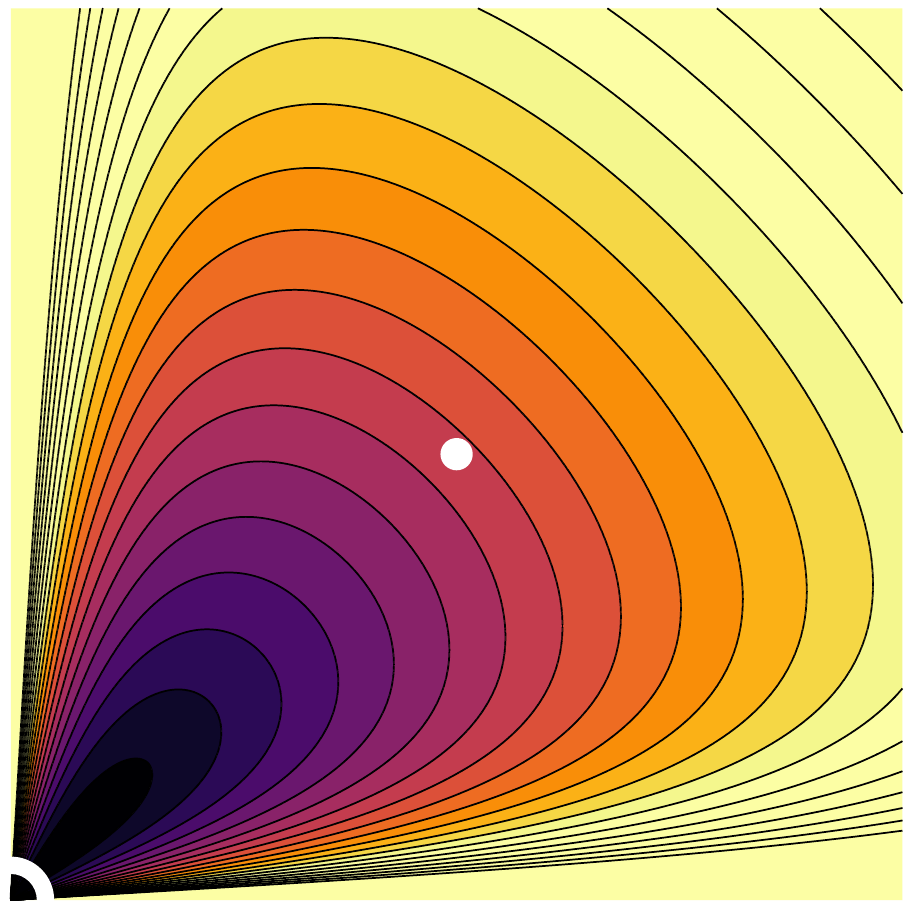}&
         \includegraphics[height=.15\linewidth]{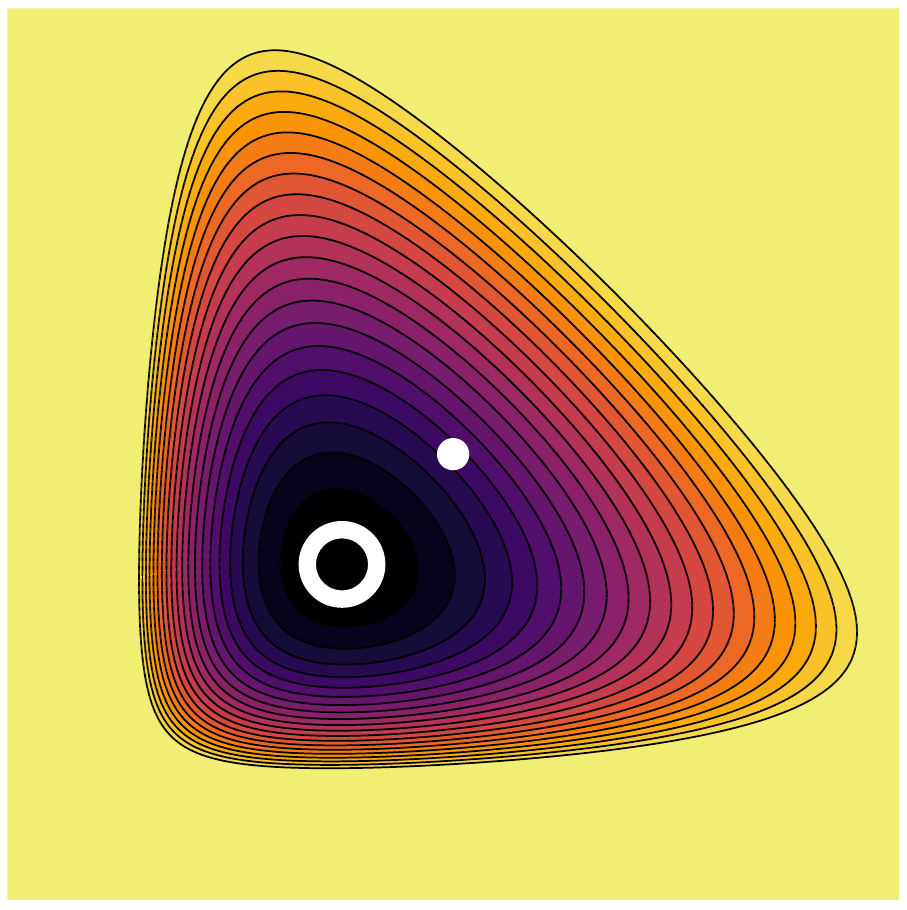}&
         \includegraphics[height=.15\linewidth]{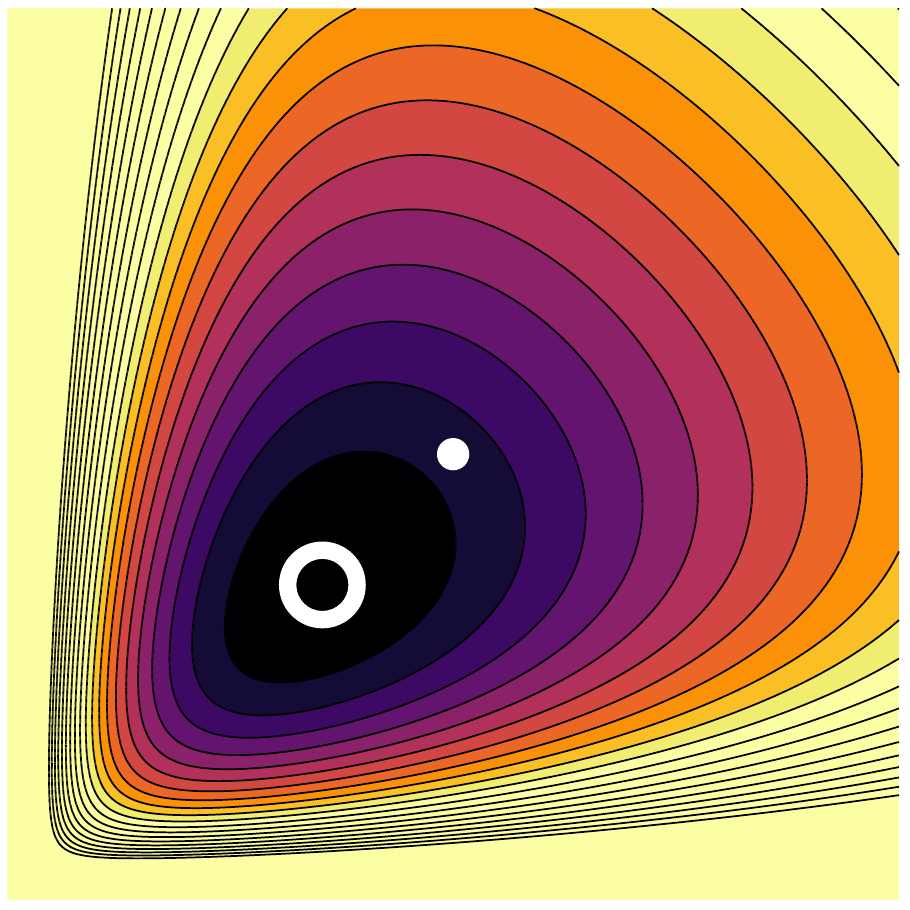}&
         \includegraphics[height=.15\linewidth]{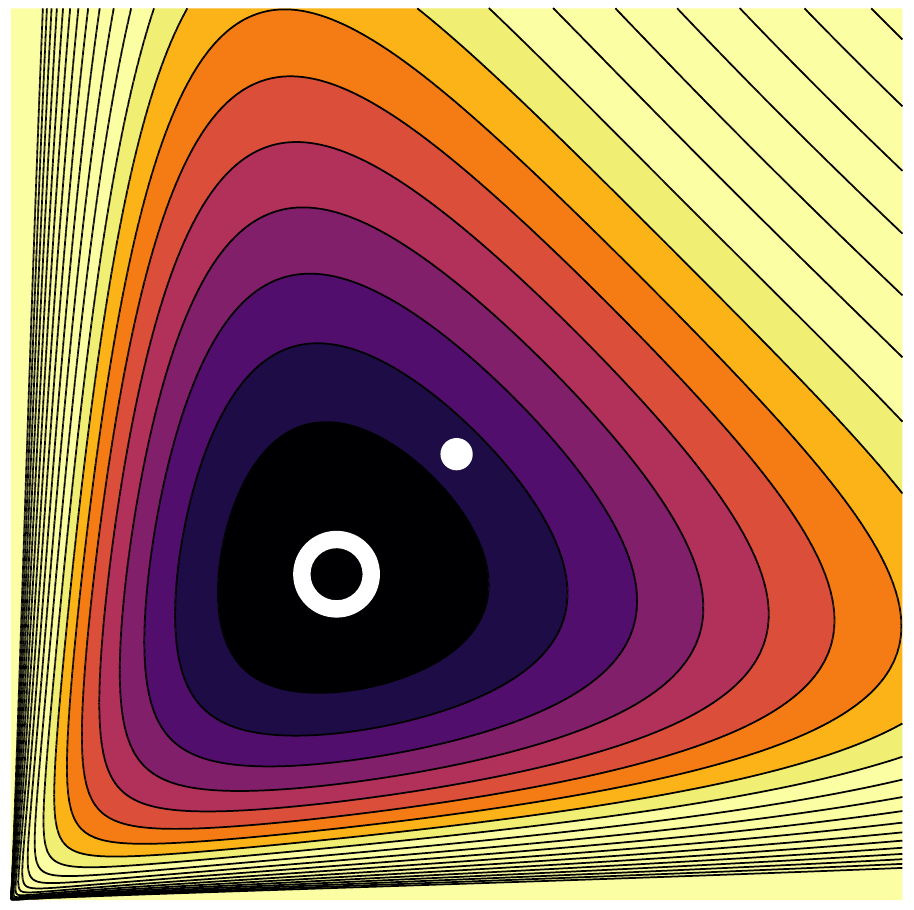}&
         \includegraphics[height=.15\linewidth]{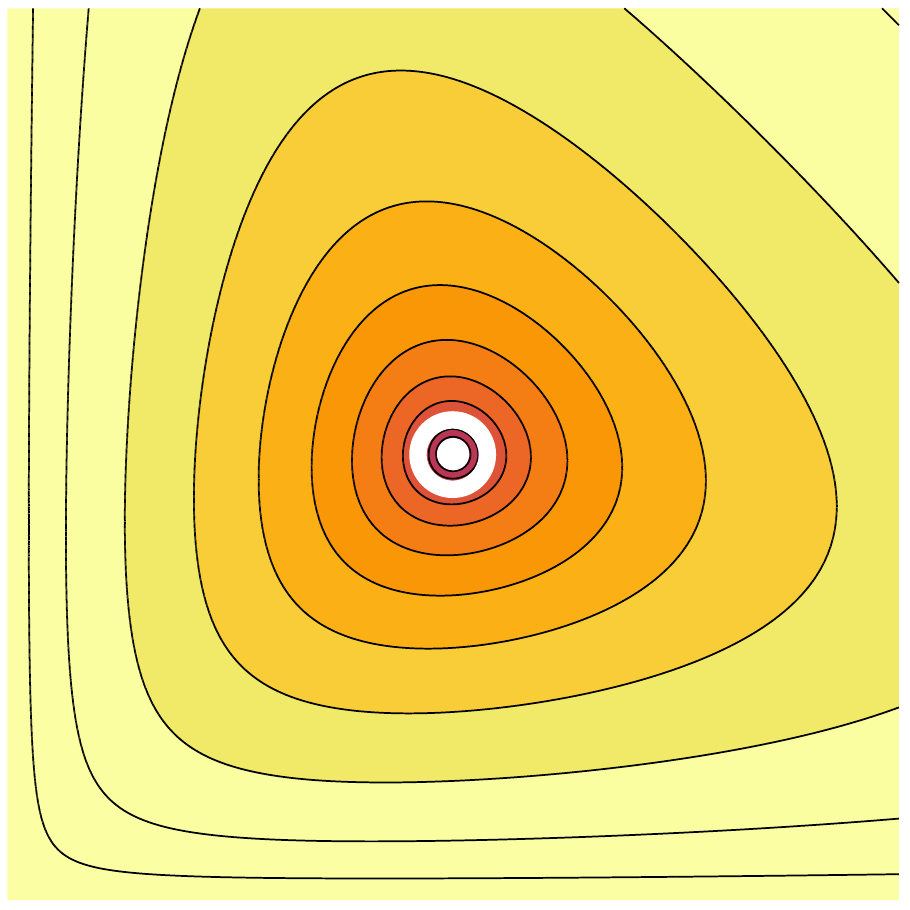}&
         \includegraphics[height=.15\linewidth]{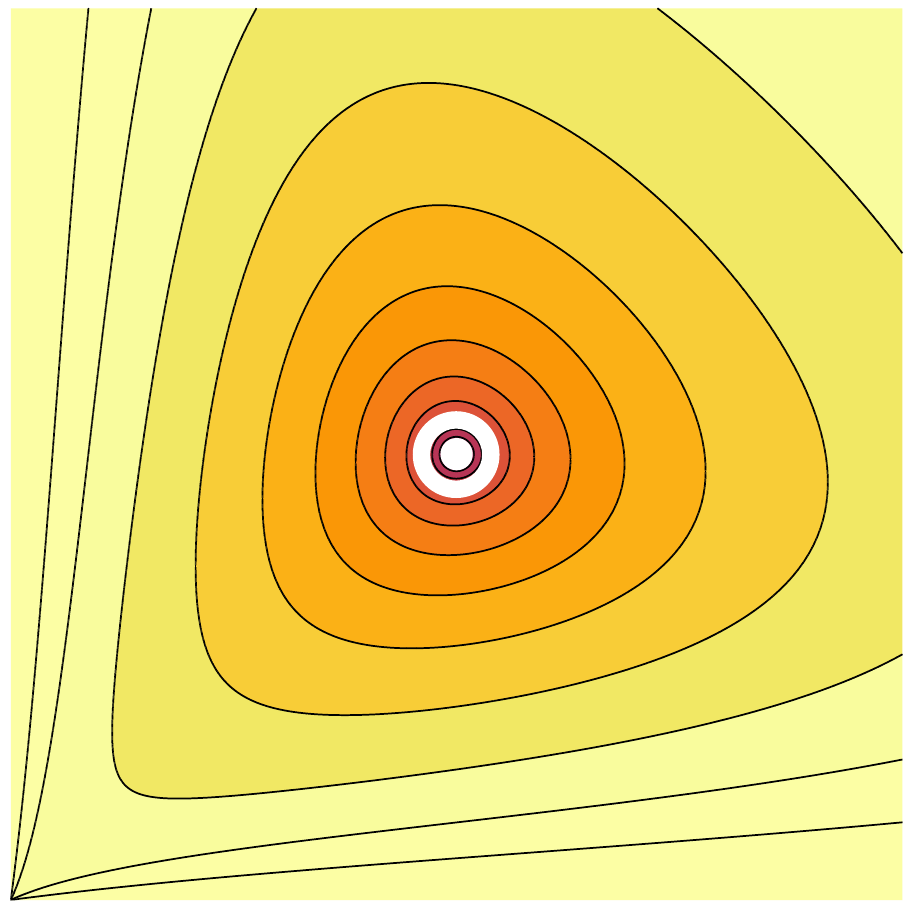}
    \end{tabular}
    \caption{
    Level sets of distortion functions $f$ (top) and their symmetrized counterparts $f^{\mathrm{Sym}}$ (bottom) evaluated at $(\sigma_1,\sigma_2,1)$ for $(\sigma_1,\sigma_2)\in[0,2]^2.$ We mark $(1,1)$ as a white dot and the location of the minimum as a circle. In the parlance of \S\ref{sec:symm-design}, all energies except the Dirichlet energy preserve structure ($f$ minimized at $(1,1,1)$), while only the Hencky strain and ARAP energies favor isometry ($f^{\mathrm{Sym}}$ minimized at $(1,1,1)$).  Only Dirichlet and ARAP are nonsingular, since the level sets do not diverge as singular values approach $0$.
    }
    \label{fig:levelsets}
\end{figure*}

\subsection{Symmetrized Energy Functions}
\label{s:symm-energy}

For correspondence problems where there is no clear distinction between the rest pose and the target pose, it is desirable for a volumetric correspondence method to be \emph{symmetric}, meaning that it is invariant to the ordering of the ``source'' domain $M_1$ and ``target'' domain $M_2$. \resub{Symmetry requires $E_f[\phi] = E_f[\phi^{-1}]$. In this section, we arrive at a set of conditions on $f$ to check if an energy is symmetric, and propose a symmetrization procedure to obtain the symmetrized form of a distortion function $f$. We later investigate the effects on computing a map using the symmetrized form of $f$.}

Following~\cite{schreiner2004inter,ezuz2019reversible,schmidt2019distortion,cachier2000symmetrization,christensen2001consistent}, one simple way to achieve symmetry is to optimize the average of the distortion energy of a map with the distortion energy of its inverse. \citet{ezuz2019reversible} and \citet{schreiner2004inter} use the simplest choice of energies to symmetrize---the Dirichlet energy---while~\citet{schmidt2019distortion} use the symmetric Dirichlet energy to prevent foldovers. Below, we analyze the consequences of using these energies and other possible choices of $f$ not considered in prior work. 
Surprisingly, our analysis will show that the Dirichlet energy and several other seemingly reasonable choices do \emph{not} yield an effective notion of distortion after symmetrization, leading us to employ an alternative in our technique.

%

We start by deriving conditions on $f$ that ensure the distortion energy $E_f$ is invariant to the ordering of the source and target.
Let $M_1$ and $M_2$ be open subsets of $\R^n$ and $\phi : M_1 \rightarrow M_2$ a diffeomorphism between them. \resub{For simplicity, assume $M_1$ and $M_2$ are normalized to have volume $1$}. 
We can compute the distortion of the map $\phi$ by applying Eq.~\eqref{eqn:map-distortion} in both directions:
\begin{align}
    E_f[\phi] &= \int_{M_1}f\left(J_{\phi}\left(\mathbf{x}\right)\right)dV_{1}(\mathbf{x})\label{eqn:m1-dist}\\
    E_f[\phi^{-1}] &= \int_{M_2}f\left(J_{\phi^{-1}}\left(\mathbf{y}\right)\right)dV_2(\mathbf{y}).\label{eqn:m2-dist}
\end{align}

Pulling back the integral in Eq.~\eqref{eqn:m2-dist} to $M_1$, we use a change of variables to $\mathbf{y} = \phi(\mathbf{u})$ to show
\begin{equation}
    E_f[\phi^{-1}] = \int_{M_1} f\left(J_{\phi^{-1}}\left(\phi\left(\mathbf{u}\right)\right)\right)\left\vert\det J_{\phi}(\mathbf{u})\right\vert \,dV_1(\mathbf{u}).
\end{equation}

By the inverse function theorem,
\begin{equation}
\label{eqn:symmetrized-int}
    E_f[\phi^{-1}] = \int_{M_1} f\left(\left(J_{\phi}\left(\mathbf{u}\right)\right)^{-1}\right) \left\vert\det J_{\phi}\left(\mathbf{u}\right)\right\vert \, dV_1(\mathbf{u}).
\end{equation}

For invariance with respect to the integration domain,  Eq.~\eqref{eqn:m1-dist} must agree with Eq.~\eqref{eqn:symmetrized-int}.
Matching the integrands, 
\begin{equation}\label{eq:invertiblecondition}
\boxed{f(J)=\left\vert\det J\right\vert f\left(J^{-1}\right)},
\end{equation}
is sufficient for this equivalence. In terms of the singular values, we obtain
\begin{equation}\label{eq:invertiblecondition_sigma}
f(\singularvalues)=\left\vert \prod_{i=1}^n \sigma_i \right\vert f\left(\frac1{\sigma_1},\ldots,\frac1{\sigma_n}\right).
\end{equation}
Here and in our subsequent discussion, we will use $n$ to refer to the dimensionality of the domains $M_1,M_2$ when the result under discussion applies to maps in any dimension; $n=3$ in our application. \resubb{This condition was first proposed by~\citet{cachier2000symmetrization} to propose symmetrization by averaging the distortion function in both mapping directions.} This motivates the following definition:

\begin{definition}[Symmetric energy]
    A distortion energy $E_f$ whose distortion function $f$ satisfies Eq.~\eqref{eq:invertiblecondition}---or Eq.~\eqref{eq:invertiblecondition_sigma} in terms of singular values---is a \emph{symmetric energy}.\label{def:symmetric}
\end{definition}

Our symmetric energy condition is both necessary and sufficient for symmetric distortion measures, in the following sense:
\begin{proposition}
$E_f[\phi]=E_f[\phi^{-1}]$ for all $M_1$, $M_2$, and $\phi$ as defined above if and only if $f$ is a symmetric energy.
\end{proposition}
\begin{proof}
Substituting \eqref{eq:invertiblecondition} into \eqref{eqn:symmetrized-int} shows that any $f$ satisfying \eqref{eq:invertiblecondition} automatically satisfies $E_f[\phi]=E_f[\phi^{-1}]$.  We now show the converse. Since $E_f[\phi]=E_f[\phi^{-1}]$ $\forall$ $M_1$, $M_2$, $\phi$ as defined above, we can choose $M_1 = B_1(\mathbf{0}) \subset \R^n$, the open ball of radius $1$. Consider any invertible $J\in\R^{n\times n}$, and define a map $\phi(\mathbf{x}) := J\mathbf{x}$, whose Jacobian is given by $J_{\phi}(\mathbf{x}) = J$. Take $M_2:= \phi(M_1)$. Applying \eqref{eqn:m1-dist},
\begin{equation}
\begin{aligned}
    \label{eqn:proof-Ef}
    E_f[\phi] &= 
    f(J) \cdot \mathrm{vol}\left(B_1(\mathbf 0)\right).
    \end{aligned}
\end{equation}
Similarly, applying \eqref{eqn:symmetrized-int} yields
\begin{equation}
    \label{eqn:proof-Er}
    E_f[\phi^{-1}] 
    = f\left( J^{-1} \right) \left \vert \det J \right \vert \cdot \mathrm{vol}\left(B_1( \mathbf{0} ) \right).
\end{equation}
Equating Eq.~\eqref{eqn:proof-Ef} and Eq.~\eqref{eqn:proof-Er} and dividing by $\mathrm{vol}\left(B_1( \mathbf{0} ) \right)$ completes the proof.
\end{proof}

Not all distortion energies are symmetric, but there is a simple procedure to construct a symmetric distortion function $f^{\mathrm{Sym}}$ from any distortion function $f$.
For any distortion function $f$, we can obtain a corresponding $f^{\mathrm{Sym}}$ fulfilling Eq.~\eqref{eq:invertiblecondition} by---in effect---computing $\frac{1}{2}E_f[\phi]+\frac{1}{2}E_f[\phi^{-1}]$ via our symmetrization procedure:
\begin{equation}\label{eqn:symmetrized}
f^{\mathrm{Sym}}(J)=\frac{1}{2}f(J) + \frac{1}{2}\left\vert\det J\right\vert f\left(J^{-1}\right),
\end{equation}
or in terms of singular values,
\begin{equation}\label{eqn:symmetrized_sing}
f^{\mathrm{Sym}}(\singularvalues)=\frac{1}{2}f(\singularvalues) + \frac{1}{2}\left\vert \prod_{i=1}^n \sigma_i \right\vert f\left(\frac1{\sigma_1},\ldots,\frac1{\sigma_n}\right).
\end{equation}

For example, suppose $f_{\mathrm{D}}(J)=\| J\|_{F}^2$ is the distortion function of the Dirichlet energy.  Then, the average of the Dirichlet energy of the forward map and of the inverse map yields the distortion function:
\begin{equation}\label{eqn:symmetrizeddirichlet}
f_{\mathrm{D}}^{\mathrm{Sym}}(J) = \frac{1}{2}\|J\|_{F}^2 + \frac{1}{2}\vert \det J \vert \|J^{-1}\|_F^2,
\end{equation}
or for $n = 3$,
\begin{equation}
    f_{\mathrm{D}}^{\mathrm{Sym}}(\sigma_1,\sigma_2,\sigma_3) =  \frac{1}{2}\sum_{i=1}^3\sigma_i^2 + \frac{1}{2}\left(\sigma_1\sigma_2\sigma_3\right) \left( \sum_{j=1}^3\sigma_j^{-2}\right)
\end{equation}
This is not the ``symmetric'' Dirichlet energy from past work on parameterization~\cite{smith2015bijective,rabinovich2016scalable}, which has the form $\frac{1}{2}\|J\|_F^2 + \frac{1}{2}\|J^{-1}\|_F^2$. Incidentally, in 2D, the second term in Eq.~\eqref{eqn:symmetrizeddirichlet} is the objective function of the inverse harmonic mapping problem used to obtain foldover-free mappings by \citet{garanzha2021foldover}. \resub{This term is also known as the inverse Dirichlet energy~\cite{knupp1995mesh}.}

Eq.\ \eqref{eqn:symmetrizeddirichlet} is a model for the objective function for mapping surfaces in~\cite{schreiner2004inter,ezuz2019reversible}, and one could reasonably attempt to reuse the same formulation for volumes.  More careful examination of this function, however, indicates some undesirable properties.  In particular, as illustrated in Fig.~\ref{fig:levelsets}, the distortion function $f_{\mathrm{D}}^{\mathrm{Sym}}(\singularvalues)$ is not minimized at $(1,1,1)$, the singular values of a rigid map.  That is, the distortion function of the symmetrized Dirichlet energy $f_{\mathrm D}^{Sym}$ favors non-isometric maps, even though it is symmetric.

The counterintuitive behavior of energies like in Eq.~\eqref{eqn:symmetrizeddirichlet} suggests that algorithms optimizing the sum of the distortion of a map and the distortion of its inverse can have unpredictable behavior, even for standard choices of distortion functions. We examine this effect empirically in \S\ref{sec:symmetrized_energy_choice}.


\subsection{Designing Symmetric Distortion Energies}
\label{sec:symm-design}
\resub{In this section, we extend the previous analysis to compute the symmetrized form of several commonly used distortion functions and examine their behavior in computing a volumetric map. We propose a list of desiderata to guide the selection of a desirable distortion function $f$.}

Several properties are desirable when selecting $f$:
\begin{itemize}
    \item Favors isometry:  $f^\mathrm{Sym}$ is minimized at $(1,1,1)$.
    \item Preserves structure:  $f$ is minimized at $(1,1,1)$.
    \item Nonsingular:  $f$ is defined for all matrices.
\end{itemize}
Favoring isometry and preserving structure are similar but not identical conditions, and they are desirable for different reasons.  Distortion energy functions that favor isometry are the typical choice for geometry processing applications, and this condition simply expresses a preference for maps $\phi$ that are rigid.  On the other hand, structure-preserving choices of $f$ facilitate optimization routines like ours that alternate between estimating $\phi$ and $\psi$, ensuring that both alternating steps work toward a common goal.  Similarly, nonsingular functions $f$ avoid the need for barrier optimization techniques and feasible initialization. 

The following proposition provides a necessary condition that can be used to rule \emph{out} many standard choices of $f$ when considering the properties above:
\begin{proposition}
Suppose a differentiable function $f:\R^3\to\R_{\geq0}$ favors isometry and preserves structure, i.e., $f(\singularvalues)$ and $f^{\mathrm{Sym}}$ are minimized at $(1,1,1)$.  Then, $f(1,1,1)=0$ and $\nabla f(1,1,1)=(0,0,0)$.
\end{proposition}
\begin{proof}
Structure preservation immediately implies $\nabla f(1,1,1)=(0,0,0)$ since $(1,1,1)$ is a local minimum.  Similarly, to favor isometry, we must have that $\nabla f^{\mathrm{Sym}}(1,1,1)=(0,0,0)$.  Taking the derivative of \eqref{eqn:symmetrized_sing} in one singular value $\sigma_i$, we find
$$
\frac{\partial f^{\mathrm{Sym}}}{\partial \sigma_i}
=
\frac12 \frac{\partial f}{\partial \sigma_i}
+ \frac12 
\left\vert
\prod_{j\neq i}\sigma_j
\right\vert
\left[
f\left(\frac1{\sigma_1},\ldots,\frac1{\sigma_n}\right)
-
\frac{1}{|\sigma_i|}\frac{\partial f}{\partial \sigma_i}\left(\frac1{\sigma_1},\ldots,\frac1{\sigma_n}\right)
\right].
$$
Substituting $\sigma_1=\cdots=\sigma_n=1$, 
$$
0=\frac{\partial f^{\mathrm{Sym}}}{\partial \sigma_i}
(1,\ldots,1)
=\frac12 
f\left(1,\ldots,1\right).
$$
This expression yields our first condition.
\end{proof}
The result above may feel somewhat counterintuitive, since constant shifts in $f$ affect whether $f$ favors isometry.  \resubb{But, adding a constant to $f$ changes the effect of the volume form on the distortion energy, explaining the result above.}

In Table~\ref{tab:energies}, we list several distortion functions $f(J)$, their equivalent forms in terms of the Jacobian $J$'s singular values $f(\singularvalues)$, and their symmetrized forms $f^{\mathrm{Sym}}(J),\, f^{\mathrm{Sym}}(\singularvalues)$. We check if the symmetrized distortion functions satisfy the isometry favoring property above by examining the behavior of $\singularvalues_{\min}$, the singular values that minimize $f^{\mathrm{Sym}}(\singularvalues)$. We verify the other properties in a similar way by studying $f(\singularvalues)$. Table~\ref{tab:energy-properties} summarizes the result.  Figure \ref{fig:levelsets} visualizes these properties by showing level sets of $f$ and $f^{\mathrm{Sym}}$ for examples drawn from Table~\ref{tab:energies}.

Tables~\ref{tab:energies} and~\ref{tab:energy-properties} reveal several valuable properties that can inform our choice of $f$. 
None of the distortion energies in Table~\ref{tab:energies} is symmetric in its standard form. A surprising result is that, after symmetrization, no distortion energy except for ARAP and Hencky strain favors isometry. Despite the fact that minimizing these energies in the forward or reverse direction independently would lead to an isometry, minimizing for the average of the two does not (see Fig.~\ref{fig:levelsets}). For example, the symmetric Dirichlet energy and the AMIPS energy after symmetrization prefer maps that tend to shrink $(\sigma_{\min}<1)$. We also observe that the symmetrized Dirichlet, \resub{the symmetrized \nth{3}-order Dirichlet}, and the symmetrized MIPS energies favor maps that collapse, that is, they are minimized close to $\mathbf{\sigma_{\mathrm{min}}}\approx (0,0,0)$. While the (asymmetric) Dirichlet energy favors maps with $\singularvalues=\mathbf 0$, the MIPS energy does not. \resub{The \nth{3}-order Dirichlet energy is used in 3D for $C^1$ continuity~\cite{iwaniec2010deformations}. }


From Table~\ref{tab:energy-properties}, only the symmetrized ARAP energy, which we will refer to as \emph{sARAP}, satisfies all the desired properties. To implement the sARAP energy, we optimize the average of the ARAP energy of the forward and reverse maps. This objective function has the added benefit of removing the requirement of a flip-free initialization, which is often not available for correspondence tasks. 

\resub{If $M_1$ and $M_2$ have different volumes, then the forward and backward terms in Eqs.~\eqref{eqn:m1-dist}, \eqref{eqn:m2-dist} might prefer distortion of one direction over another. In practice, we normalize our models to have volume $1$, so that the integrals in Eqs.~\eqref{eqn:m1-dist},\eqref{eqn:m2-dist} measure average local distortion of the two maps; \citet{schreiner2004inter} equivalently rescales the forward and backward terms.}



\begin{figure}
    \centering
    \begin{tabular}{cc}
    \includegraphics[height=.35\linewidth]{figures/energyplots/arap_symm.pdf} &
    \includegraphics[height=.35\linewidth]{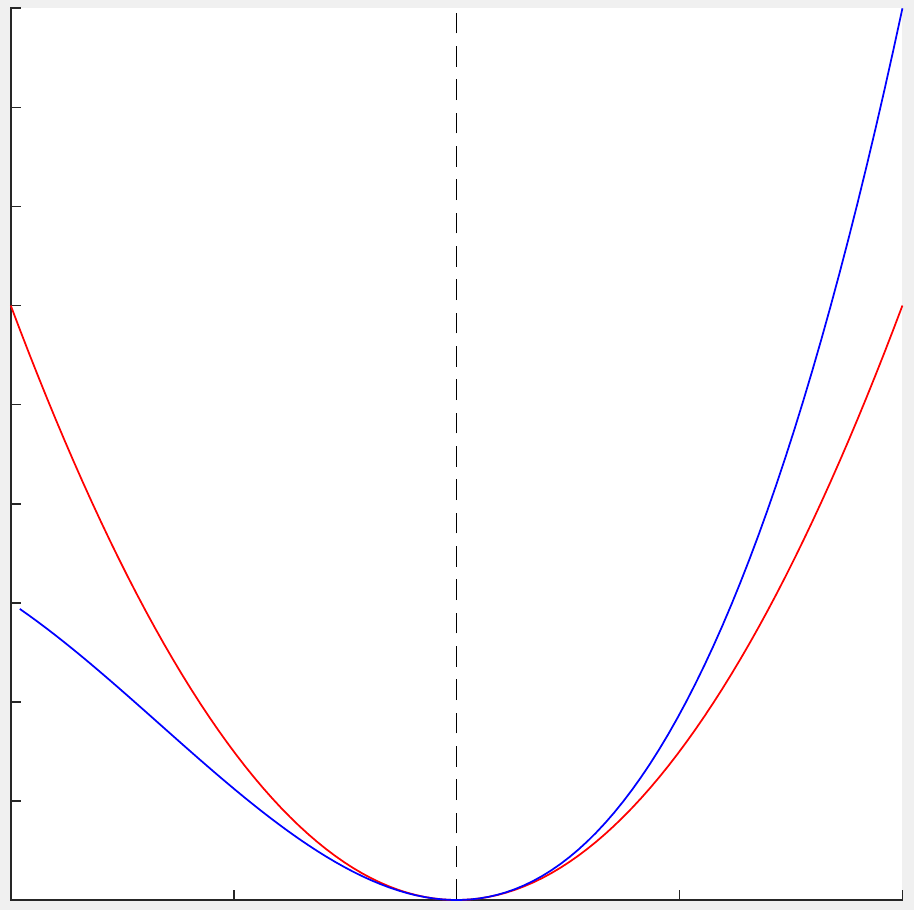}
    \\
    \includegraphics[height=.35\linewidth]{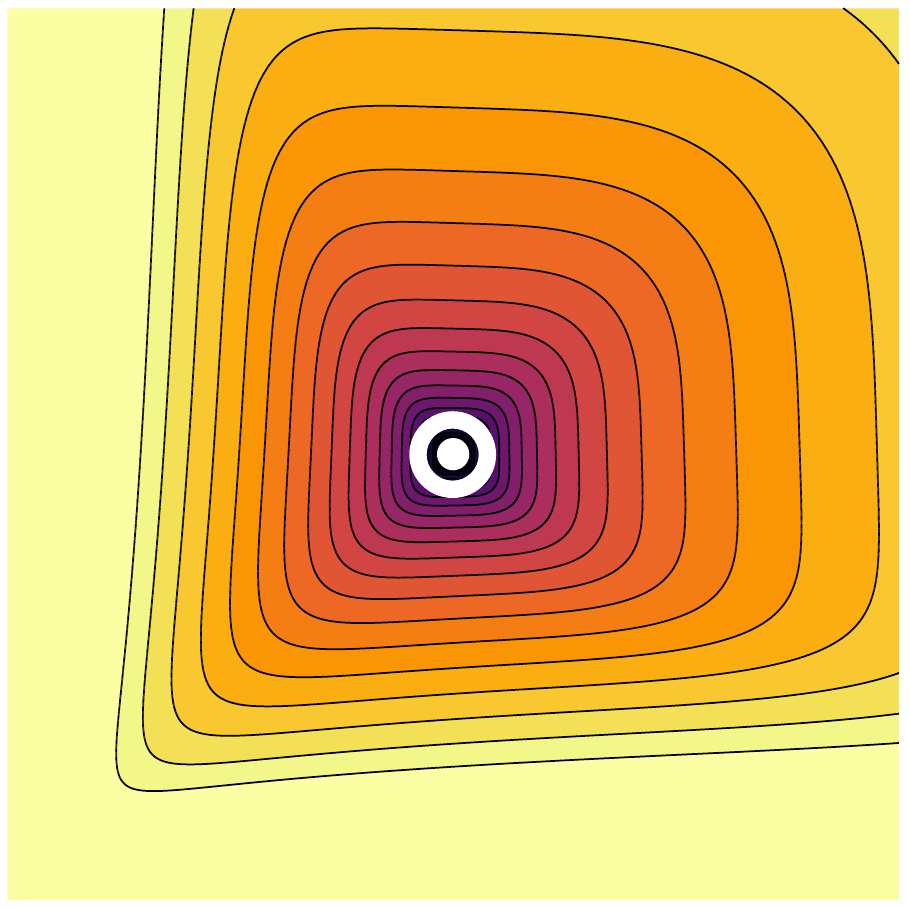} & 
    \includegraphics[height=.35\linewidth]{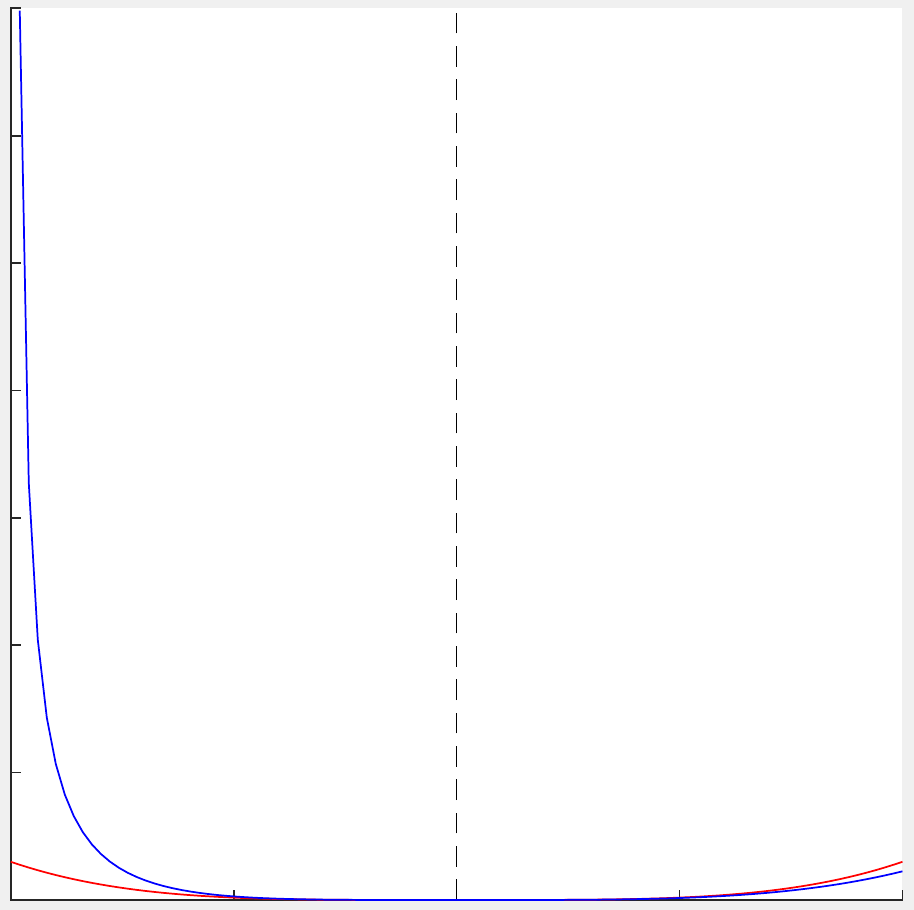}
    \\
    $f^{\mathrm{Sym}}(\sigma_1,\sigma_2,1)$ & {\color{red}$f(\sigma,\sigma,\sigma)$} and {\color{blue}$f^{\mathrm{Sym}}(\sigma,\sigma,\sigma)$}
    \end{tabular}
    \caption{Mathematical boundary case:  Comparison of symmetrized ARAP energy $\sum_i(\sigma_i-1)^2$ to symmetrized fourth-power ARAP energy $\sum_i(\sigma_i-1)^4$, using level sets similar to Figure \protect\ref{fig:levelsets} (left) and by plotting the diagonal where $\sigma=\sigma_1=\sigma_2=\sigma_3$ (right).  As discussed in \S\ref{sec:symm-design} (Remark), the fourth-power alternative blows up when approaching $(0,0,0)$ from \emph{any} direction, while conventional ARAP admits a path to $(0,0,0)$ where the energy density remains finite.}
    \label{fig:boundarycase}
\end{figure}

\begin{remark}[Avoiding zero singular values]
The symmetric Dirichlet energy \cite{smith2015bijective}, symmetric gradient energy~\cite{stein2021splitting}, and others used for bijective parameterization \emph{blow up} as singular values approach zero; this property provides a barrier ensuring existence of a locally-optimal parameterization without collapsed or inverted elements.  Our \emph{nonsingular} property actually prefers the opposite of this scenario, allowing inverted Jacobians so that we can recover from poor initialization, but this is a property of $f$---employed during optimization---rather than $f^{\mathrm{Sym}}$, the actual distortion energy being optimized in the symmetrized formulation.

A nonsingular $f$ can actually admit a function $f^{\mathrm{Sym}}$ that blows up as singular values approach $0$, as is the case for the ARAP and Dirichlet energies. This property suggests that even a nonsingular choice of $f$ can favor orientation-preserving symmetric maps.

For completeness, we note that $f_{\mathrm{ARAP}}^{\mathrm{Sym}}$ is not a \emph{perfect} barrier, in the following sense (also illustrated in Figure \ref{fig:boundarycase}):
For $\sigma_1=1$ and $\sigma_2, \sigma_3 \rightarrow 0$, we have $f_{\mathrm{ARAP}}^{\mathrm{Sym}}(\singularvalues) \rightarrow 1$.
This technicality can be addressed using an $f$ that grows faster than cubically in the singular values, e.g.\ $f(\singularvalues)=\sum_i(\sigma_i-1)^4$, but in practice such an adjustment did not yield better maps.
\end{remark}

\begin{remark}[Role of boundary conditions]
Several prior works \resub{optimize} symmetric energies \emph{without} the desired properties at the beginning of this section~\cite{ezuz2019reversible,schreiner2004inter,schmidt2019distortion}. 
Although their distortion energies do not promote isometry directly, these methods are still able to find low-distortion and even bijective correspondences. Indeed, the symmetrized energy analysis above does not tell the whole story. In particular, these methods \resub{include} energy terms, boundary conditions, and other constraints that favor bijectivity and semantic correspondences. \resub{These constraints counteract the energy's unexpected local properties and can affect the resulting map quality. For example, optimizing the symmetrized Dirichlet energy in the space of surjective or bijective maps will prevent the map from collapsing, but the map quality is essentially upheld by the boundary condition rather than the constitutive model used in the objective function.} We hypothesize that the success of these methods lies in balancing \resub{competing} terms and constraints. 
We leave detailed theoretical analysis of these intriguing global questions to future work.
\end{remark}


\subsection{Symmetric Optimization Problem}
\resub{Following the previous section's analysis, we revise the the generic formulation of our optimization problem in Eq.~\eqref{eq:opt-problem} to be symmetric. We optimize an energy of the form $\frac12E_f[\phi]+\frac12E_f[\psi]$, where we maintain separate estimates of the map $\phi:M_1\to M_2$ and its inverse $\psi\approx\phi^{-1}:M_2\to M_1$. This is done for practical reasons: The existence of a flip-free initial map is not guaranteed, so $\phi^{-1}$ may not exist to start. Additionally, this form is advantageous as $f$ is necessarily nonsingular for initializations with flipped elements, while $f^{\textrm{Sym}}$ can be orientation-preserving as is the case for sARAP. Finally, even if $f$ is not symmetric, the resulting energy is roughly of the form in Eq.~\eqref{eqn:symmetrized} and hence our analysis in \S\ref{s:symm-energy} applies. This leads to the modified problem:}

 \begin{equation}
\label{eq:opt-problem-sym}
\begin{aligned}
\argmin_{\phi,\psi} \quad &      \frac{1}{2}\int_{M_1} f_{\mathrm{ARAP}}(J_{\phi}(\mathbf{x}) )\, dV(\mathbf{x}) \\ 
&+ \frac{1}{2}\int_{M_2} f_{\mathrm{ARAP}}(J_{\psi}(\mathbf{y}) )\, dV(\mathbf{y}) + \mathrm{Reg}[\phi,\psi]\\
\textrm{subject to} \quad & \phi \in \mathcal{P} \, , \psi \in \mathcal{Q},\\
\end{aligned}
\end{equation}
where $\mathcal{Q}$ denotes the constraint $\psi(\partial M_2) \subset \partial M_1$. In practice, the constraints that define $\mathcal{P}$ and $\mathcal{Q}$ can be made soft and modeled in $\mathrm{Reg}[\phi,\psi]$. The estimate $\psi \approx \phi^{-1}$ can be enforced as a soft or hard constraint. In practice, we use a soft constraint modeled in $\mathrm{Reg}[\phi,\psi]$ as described in \S\ref{sec:objectiveterms}.



\begin{table*}
  \caption{Several distortion measures and their symmetrized forms. In this table, we consider orientation-preserving maps, so that $\vert \det J \vert = \det J$. We use an interior-point method constrained to search over non-negative $\sigma$ to compute the set of singular values $\singularvalues_{\min}$ that minimize the symmetrized energy $f^{\mathrm{Sym}}(\singularvalues)$.}
  \label{tab:energies}
  \scriptsize
  \begin{tabular}[0.9\textwidth]{|l|ccccc|}
    \hline
    Name & $f(J)$ & $f(\mathbf \sigma)$ & $f^{\mathrm{Sym}}(J)$ & $f^{\mathrm{Sym}}(\singularvalues$) & $\singularvalues_{\min}$\\ 
    \hline
    Dirichlet  & $\|J\|_F^2$ & $\sum_{i=1}^n \sigma_i^2$ 
    & $\frac{1}{2}\|J\|_F^2 + \frac{1}{2}\left( \det J \right) \left(\|J^{-1}\|_F^2\right)$ 
    & $\makecell{ \frac{1}{2}\sum_{i=1}^n\sigma_i^2 \\ + \frac{1}{2}\left(\prod_{j=1}^n\sigma_j\right) \left( \sum_{k=1}^n \sigma_k^{-2} \right)}$
    & $\approx (0,0,0)$\\ 
    \hline
    \resub{Dirichlet (\nth{3} order)}  & \resub{$\|J\|_F^3$ } & \resub{$\sum_{i=1}^3 \sigma_i^3$ }
    & \resub{$\frac{1}{2}\|J\|_F^3 + \frac{1}{2}\left( \det J \right) \left(\|J^{-1}\|_F^3\right)$} 
    & \resub{$\makecell{ \frac{1}{2}\sum_{i=1}^3\sigma_i^3 \\ + \frac{1}{2}\left(\prod_{j=1}^3\sigma_j\right) \left( \sum_{k=1}^3 \sigma_k^{-3} \right)}$}
    & \resub{$\approx (0,0,0)$}\\ 
    \hline
    Symmetric Dirichlet & $\|J\|_F^2 + \|J^{-1}\|_F^2$ &
    $\sum_{i=1}^n \left( \sigma_i^2 + \sigma_i^{-2} \right)$
    & $\makecell{\frac{1}{2}\left(\det J \, + 1\right) \left( \|J\|_F^2 + \|J^{-1}\|_F^2 \right)}$
    & $\makecell{ \frac{1}{2}\left(\prod_{i=1}^n\sigma_i \, + 1\right) \left( \sum_{j=1}^n \left( \sigma_j^2+\sigma_j^{-2} \right) \right)}$
    & $\approx(0.77,0.77,0.77)$ \\
    \hline
    MIPS (3D) & $\frac{1}{8}\left(\|J\|_F^2\cdot\|J^{-1}\|_F^2-1\right)$ & $\frac{1}{8} \prod_{i=1}^3 \left( \frac{\sigma_i}{\sigma_{i+1}}+\frac{\sigma_{i+1}}{\sigma_i}\right)$ 
    & $\makecell{ \frac{1}{16} \left(\det J \, + 1 \right) \left(\|J\|_F^2\cdot\|J^{-1}\|_F^2-1\right) }$
    &$\makecell{\frac{1}{16} \left( 1 +\prod_{i=1}^3 \sigma_i\right) \left ( \prod_{j=1}^3 \left( \frac{\sigma_{j+1}}{\sigma_{j}}+\frac{\sigma_{j}}{\sigma_{j+1}}\right) \right)}$
    & $\approx (0,0,0)$\\
    \hline
    AMIPS (3D)
    & $\makecell{\frac{1}{16}\left(\|J\|_F^2\cdot\|J^{-1}\|_F^2-1\right) \\ + \frac{1}{2}\left(\det J + \left(\det J\right)^{-1}\right)}$
    
    & $\makecell{\frac{1}{16} \prod_{i=1}^3 \left( \frac{\sigma_i}{\sigma_{i+1}}+\frac{\sigma_{i+1}}{\sigma_i}\right)\\
    + \frac{1}{2}\left(\prod_{j=1}^3 \sigma_j + \prod_{k=1}^n\sigma^{-1}_k \right)}$
    
    & $\makecell{\frac{\det J \, + 1}{32}\left(\|J\|_F^2\cdot\|J^{-1}\|_F^2-1\right) \\ + \frac{1}{4}\left (\det J + \left(\det J\right)^{-1} \right)\\
      + \frac{1}{4}\left( \left(\det J\right)^2 + 1\right )}$
    
    & $\makecell{\frac{1}{32} \left( 1 + \prod_{i=1}^3 \sigma_i \right) \left(  \prod_{j=1}^n \left( \frac{\sigma_j}{\sigma_{j+1}}+\frac{\sigma_{j+1}}{\sigma_j}\right) \right) \\
    + \frac{1}{4}\left(\prod_{k=1}^3 \sigma_k + \prod_{l=1}^n\sigma^{-1}_l\right) \\
    + \frac{1}{4}\left(\prod_{m=1}^3\sigma_m^2 \right)}$
    & $\approx (0.8,0.8,0.8)$ \\
    \hline
    Conformal AMIPS & $\frac{\mathrm{tr}\left(J^TJ\right)}{\left(\det J\right)^{\frac{2}{3}}} $
    
    &$\left( \prod_{j=1}^3\sigma_j^{-\frac{2}{3}}\right)\left(\sum_{i=1}^3\sigma_i^2 \right)$
    
    &  $\makecell{\frac{1}{2}\left(\det J\right)^{-\frac{2}{3}}\mathrm{tr}\left(J^TJ\right) \\
   +  \frac{1}{2} \left( \det J \right)^{\frac{1}{3}}\mathrm{tr}\left(J^{-T}J^{-1}\right)}$
    & $\makecell{\frac{1}{2}\left(\prod_{i=1}^3\sigma_i^{-\frac{2}{3}}\right)\left(\sum_{j=1}^3\sigma_j^2 \right)\\
    +\frac{1}{2} \left( \prod_{k=1}^3\sigma_k^{-\frac{1}{2}} \right) \left( \sum_{l=1}^3\sigma_l^{-2} \right)  }$
    &  $\approx (0.032,0.032,0.032)$\\
    \hline
    Symmetric 
    gradient & $\frac{1}{2}\|J\|_F^2 - \log\left(\det J\right) $ 
    &$\frac{1}{2}\sum_{j=1}^n\sigma_j^2 - \log\left(\prod_{i=1}^n\sigma_i\right)$ 
    
    &$\makecell{\frac{1}{4}\|J\|_F^2 - \frac{1}{2}\log\left(\det J\right) \\
    +\frac{1}{4}\det J \cdot \|J^{-1}\|_F^2 \\
    + \frac{1}{8}\det J \cdot \log\left( \det J\right)}$
    
    & $\makecell{\frac{1}{4}\sum_{i=1}^n\sigma_i^2  - \frac{1}{2} \log\left(\prod_{j=1}^n \sigma_j \right) \\
    +\frac{1}{4}\left( \prod_{k=1}^n\sigma_k \right) \Big [\sum_{l=1}^n\sigma_l^{-2} \\
    + \frac{1}{2} \log \left( \prod_{m=1}^n\sigma_m \right) \Big ]}$
    & $\approx(0.61,0.61,0.61)$\\
    \hline
    Hencky strain
    &$\makecell{\|\log J^T J\|_F^2}$
    & $\sum_{i=1}^n\log^2(\sigma_i)$
    & $\makecell{\frac{1}{2}\|\log J^T J\|_F^2\\
    +\frac{1}{2}\det J \cdot \|\log J^{-T}J^{-1}\|_F^2}$
    &  $\makecell{\frac{1}{2}\sum_{i=1}^n\log^2(\sigma_i)\\
    + \frac{1}{2}\left( \prod_{j=1}^n\sigma_j \right) \left( \sum_{k=1}^n\log^2(\sigma_k) \right)}$
    & $(1,1,1)$\\
        \hline
    ARAP &  $\|J-R\|_F^2$ & $\sum_{i=1}^n (\sigma_i-1)^2$ & 
    $\makecell{\frac{1}{2}\|J-R\|_F^2 \\
    + \frac{1}{2} \det J \cdot \|J^{-1}-R\|_F^2}$
    &$\makecell{\frac{1}{2}\sum_{i=1}^n (\sigma_i-1)^2 + \\
    \frac{1}{2}\left( \prod_{j=1}^n\sigma_j \right) \left(\sum_{k=1}^n (\sigma_k^{-1}-1)^2\right)}$
    &  $(1,1,1)$\\
    \hline
  \end{tabular}
\end{table*}

\section{Discretization and Model}
We build on our analysis in \S\ref{s:symm-energy} and \S\ref{sec:symm-design} to discretize the optimization problem in Eq.~\eqref{eq:opt-problem-sym} and develop an algorithm to compute a volumetric map that is invariant to the ordering of the source and target shapes. \resubb{In this section, we define our map discretization \resub{and map constraints, and develop the objective function used in the optimization.}}

\subsection{Notation}

We represent volumetric shapes as tetrahedral meshes. We let $\mathcal{V}_i$, $\mathcal{E}_i$, $\mathcal{F}_i$, $\mathcal{T}_i$ denote the sets of vertices, edges, faces, and tetrahedra of mesh $M_i$, for $i \in \{1,2\}$. We represent the coordinates of $\mathcal{V}_i$ as a matrix $V_i \in \R^{n_i \times 3}$, where $n_i$ denotes the number of vertices in mesh $M_i$. We represent tetrahedron $k$ in mesh $i$ as the matrix $V_i^{T_k} \in \R^{4\times 3}$  whose rows are the coordinates of the vertices of tetrahedron $k$. We use $\partial$ to denote the boundary of a mesh, and $\partial \mathcal{V}_i, \partial \mathcal{E}_i, \partial \mathcal{F}_i, \partial \mathcal{T}_i$ denote sets of boundary vertices, edges, faces, and tetrahedra, respectively. Boundary tetrahedra are those that contain one or more boundary faces.

\begin{table}[t]
  \caption{Summary of distortion energy function properties\vspace{-.15in}}
  \label{tab:energy-properties}
  \small
  \begin{tabular}{|l|ccc|}
    \hline
    Name & \makecell{Favors\\isometry} & \makecell{Preserves\\structure} & \makecell{Nonsingular\\ \ } \\
    \hline
    Dirichlet & \xmark & \xmark & \cmark\\
    \resub{Dirichlet (\nth{3} order)} & \xmark & \xmark & \cmark \\
    Symm.\ Dirichlet & \xmark & \cmark & \xmark \\
    MIPS (3D) &  \xmark & \cmark & \xmark  \\
    AMIPS (3D) & \xmark & \cmark & \xmark \\
    Conformal AMIPS & \xmark & \cmark & \xmark \\
    Symm.\ Gradient & \xmark & \cmark & \xmark \\
    Hencky strain & \cmark & \cmark & \xmark \\
    \textbf{ARAP} & \cmark & \cmark & \cmark \\
    \hline
  \end{tabular}
\end{table}


\resubb{We use a piecewise linear discretization to model the maps $\phi$ and $\psi$, with each tetrahedron being mapped affinely. }
The map on each tetrahedron is determined by its transformed vertex coordinates.
We use matrix $X_i \in \R^{n_1 \times 3}$ to denote the coordinates of the transformed vertices of mesh $M_i$, and $X_i^{T_k} \in \R^{4 \times 3}$ to denote the transformed tetrahedron $k$ of mesh $M_i$. The Jacobian matrix
\begin{equation}
    \label{eqn:jacobian}
    J(X_i^{T_k}) = \left(BX_i^{T_k}\right)\left(BV_i^{T_k}\right)^{-1}
\end{equation}
\resubb{defines the map differential of tetrahedron $k$ based on the transformed coordinates $X_i^{T_k}$.} The constant matrix $B \in \R^{3 \times 4}$ extracts vectors parallel to the edges of the tetrahedron. 

\subsection{Map Representation}

We wish to constrain each map to lie within the target shape, i.e., $\phi(M_1)\subset M_2$ and $\psi(M_2) \subset M_1$. We extend the strategy of~\citet{ezuz2019reversible} to tetrahedral meshes to enforce these constraints. 

We represent the map $\phi$ as a matrix $P_{12}\in [0,1]^{n_1 \times n_2}$ and the map $\psi$ as $P_{21} \in [0,1]^{n_2 \times n_1}$. \resub{Matrices $P_{12}$ and $P_{21}$ use barycentric coordinates to encode the vertex-to-tetrahedron map and ensure the mapped vertices lie in the target mesh.} \resubb{This representation is also beneficial to map between meshes with differing connectivity. } 
 Suppose $P_{12}$ maps vertex $i$ of mesh $M_1$ into tetrahedron $T_k=(a,b,c,d)\in \mathcal T_2$ in mesh $M_2$,
where $(a,b,c,d) \in \{1,\ldots,n_2\}$ are the indices of the vertices of $T_k$. 
Then, row $i$ of $P_{12}$ contains the barycentric coordinates of the image of vertex $i$ in columns $a,b,c,d$, and zeros elsewhere. Map $P_{21}$ is constructed analogously. 
We can enforce the constraint that boundary vertices are mapped to boundary faces by constraining the sparsity patterns of $P_{12}$ and $P_{21}$. \resubb{A limitation in the discretization is that we are unable to enforce that the interior of boundary faces and edges are mapped inside the target shape, since our map representation is vertex-based. In practice, this effect is minimized using high-resolution meshes}.

We denote the set of all feasible maps satisfying the boundary constraints as $\mathcal P_{ij}^\star$; we use $\mathcal P_{ij}$ to denote the set of feasible maps that may map the boundary $\partial M_i$ to the interior of $M_j$.


We use half-quadratic splitting~\cite{geman1995half} to express our problem in a form that is amenable to efficient optimization~\cite{ezuz2019reversible,wang2008new,zoran2011learning}. In particular, we introduce the auxiliary variable $X_{ij}$ to model the image of vertices $\mathcal{V}_i$ under the map to mesh $M_j$, where $X_{ij} \approx P_{ij} V_j$. 

\subsection{Objective Terms}\label{sec:objectiveterms}
\resubb{We define several objective terms used to find the correspondence and model the soft constraints on the map.}

\subsubsection{Auxiliary and reversibility energy functions.} Our first two terms are adapted from~\citet{ezuz2019reversible} and extended for volumetric meshes. The first term is the auxiliary energy that encourages $X_{ij} \approx P_{ij}V_j:$
\begin{equation}
    E_Q[P_{12}, P_{21}, X_{12}, X_{21}] = \sum_{\substack{i,j \in \{1,2\} \\  i\neq j }} \frac{1}{c_i c_j}\left \|X_{ij} - P_{ij}V_j \right \|_{M_i}^2,
\end{equation}
where $c_i$, $c_j$ are the total volumes of meshes $M_i$ and $M_j$, and ${\|\cdot\|_{M_i}^2}$ denotes the Frobenius norm with respect to $M_i$. For a matrix $G$, $\|G\|_{M_i}^2 = \mathrm{tr}(G^TC_iG)$, where $C_i$ is the lumped diagonal vertex mass matrix of $M_i$.

The second term is the reversibility energy that encourages bijectivity:
\begin{equation}
\label{reversibility-energy}
    \begin{aligned}
    E_R[P_{12},P_{21},X_{12},X_{21}] &= \sum_{\substack{i,j \in \{1,2\} \\  i\neq j }}\frac{1}{c_i^2}\|P_{ij}X_{ji} - V_i\|_{M_i}^2.
    \end{aligned}
\end{equation}
This energy measures the distance between the original vertex positions $V_i$ and the back projection of their image under the map $P_{ij}$, $X_{ij}$.

\begin{figure}[b]
    \centering
    \begin{tabular}{cc}
    \includegraphics[width=.45\linewidth]{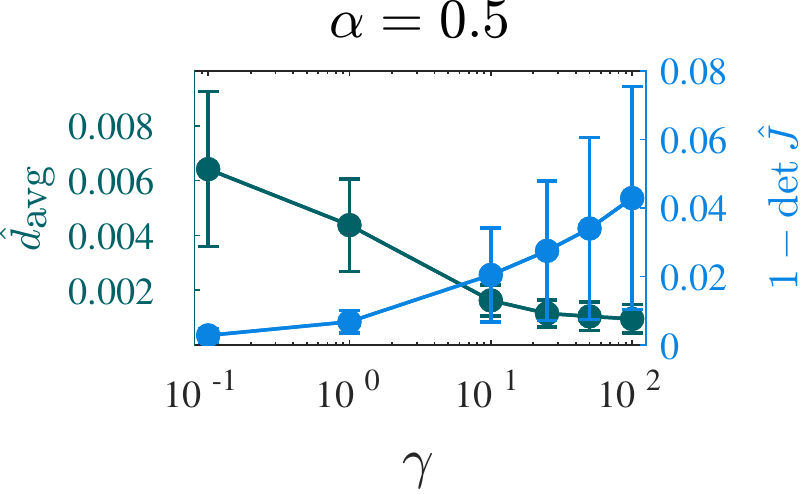} &
    \includegraphics[width=.45\linewidth]{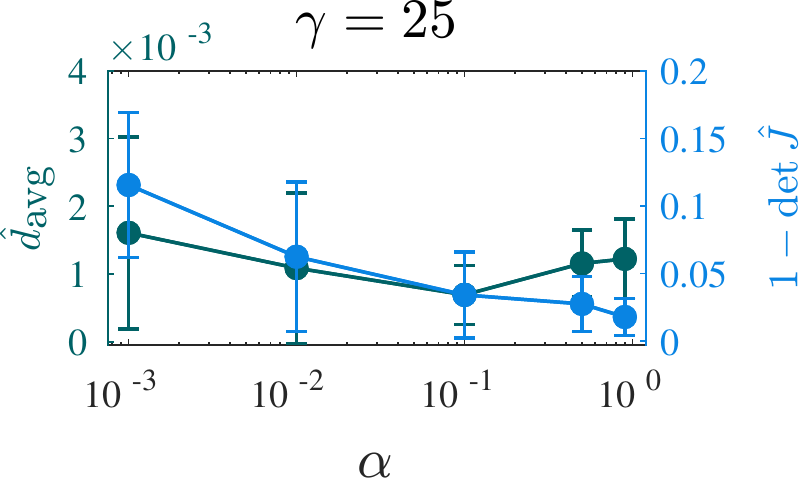}
    \\
    {\captionsize search over $\gamma$} &{\captionsize search over $\alpha$}
    \end{tabular}\vspace{-.1in}
    \caption{Parameter sweep over $\gamma$ and $\alpha$, comparing the tradeoff between $\hat{d}_{avg}$ and \resub{$1-\det\hat{J}$, where $\det\hat{J}$ is the normalized determinant of the Jacobian}. We select $\alpha=0.5$, $\gamma=25$ as they achieve a reasonable tradeoff between conforming to the target boundary while \resub{maintaining map quality}.}
    \label{fig:paramsweep}
\end{figure}

\subsubsection{ARAP energy}
Central to the computation of a volumetric map is the proper selection of a distortion energy. From our analysis in \S\ref{sec:symm-design}, we select the sARAP energy as it is both symmetric and promotes rigidity.

We use $\frac{1}{2} E_{ARAP}[\phi] + \frac{1}{2} E_{ARAP}[\psi]$ to approximate $E_{sARAP}[\phi]$. We approximate the integral over the volumetric domain by measuring the distortion energy per tetrahedron. For tetrahedron $k$ of mesh $i$, the ARAP distortion function is given by
\begin{equation}
\label{eqn:arap-energy}
\begin{aligned}
    f_{ARAP}\left(J\left(X_{ij}^{T_k}\right)\right) 
    &= \sum_{j=1}^3 (\sigma_{k,j} - 1)^2,\\
\end{aligned}
\end{equation}
where $\sigma_{k,j}$ is the $j^{th}$ \emph{signed} singular value of $J(X_{ij}^{T_k})$. We use the convention laid out by~\citet{irving2004invertible} to define the signed singular value decomposition unambiguously. For $J=U\Sigma V^T$, this convention allows the sign of the smallest singular value $\sigma_{min}$ to be negative, $\mathrm{sign}(\sigma_{min})=\mathrm{sign}(\det J)$, and $U,V \in \mathrm{SO}(3)$.

The total ARAP energy is then
\begin{equation}
    \label{eqn:total-arap}
    \begin{aligned}
    E_{ARAP}\left[X_{12},X_{21}\right] = \sum_{\substack{i,j \in \{1,2\} \\  i\neq j }} \frac{1}{2c_i}\sum_{T_k \in \mathcal{T}_i}v(T_k)f_{ARAP} \left(J\left(X_{ij}\right)^{T_k}\right),
    \end{aligned}
\end{equation}
where $v(T_k)$ denotes the volume of tetrahedron $k$.

\subsubsection{Projection Energy}
We encourage preserving the boundary of the source and target meshes by using forward and backward projection energies. We compute the forward projection energy $E_{P,f}$ as
\begin{equation}
    \label{eqn:proj-fwd}
    E_{P,f}[X_{12},X_{21}] = \sum_{\substack{i,j \in \{1,2\} \\  i\neq j }} \frac{1}{s_i} \left \|\left(X_{ij}\right)_{\partial M_i}-\mathrm{proj}\left((X_{ij})_{\partial M_i},\partial M_j\right) \right\|_{\partial M_i}^2,
\end{equation}
where $\mathrm{proj}\left(\left(X_{ij}\right)_{\partial M_i},\partial M_j\right)$ denotes the Euclidean projection of the boundary vertices of $\partial M_i$ with coordinates $X_{ij}$ onto the boundary mesh $\partial M_j$, $s_i$ denotes the total surface area of $\partial M_1$ and $\|\cdot\|_{\partial M_i}^2$ denotes the Frobenius norm with respect to boundary triangle mesh $\partial M_i$.

The backward projection energy $E_{P,b}$ is given by
\begin{equation}
    \label{eqn:proj-rev}
    E_{P,b}[X_{12},X_{21}] = \sum_{\substack{i,j \in \{1,2\} \\  i\neq j }} \frac{1}{s_i} \left \|V_i-\mathrm{proj}\left(V_i,\partial F_j \left( X_{ji}\right)\right) \right\|_{\partial M_i}^2,
\end{equation}
where $\partial F_j \left( X_{ji}\right)$ denotes the boundary of mesh $M_j$ with vertices given by $X_{ji}$.

The full projection energy is then
\begin{equation}
    E_P[X_{12},X_{21}] = E_{P,f}[X_{12},X_{21}] + E_{P,b}[X_{12},X_{21}].
\end{equation}

\begin{figure*}
    \includegraphics[width=0.9\textwidth]{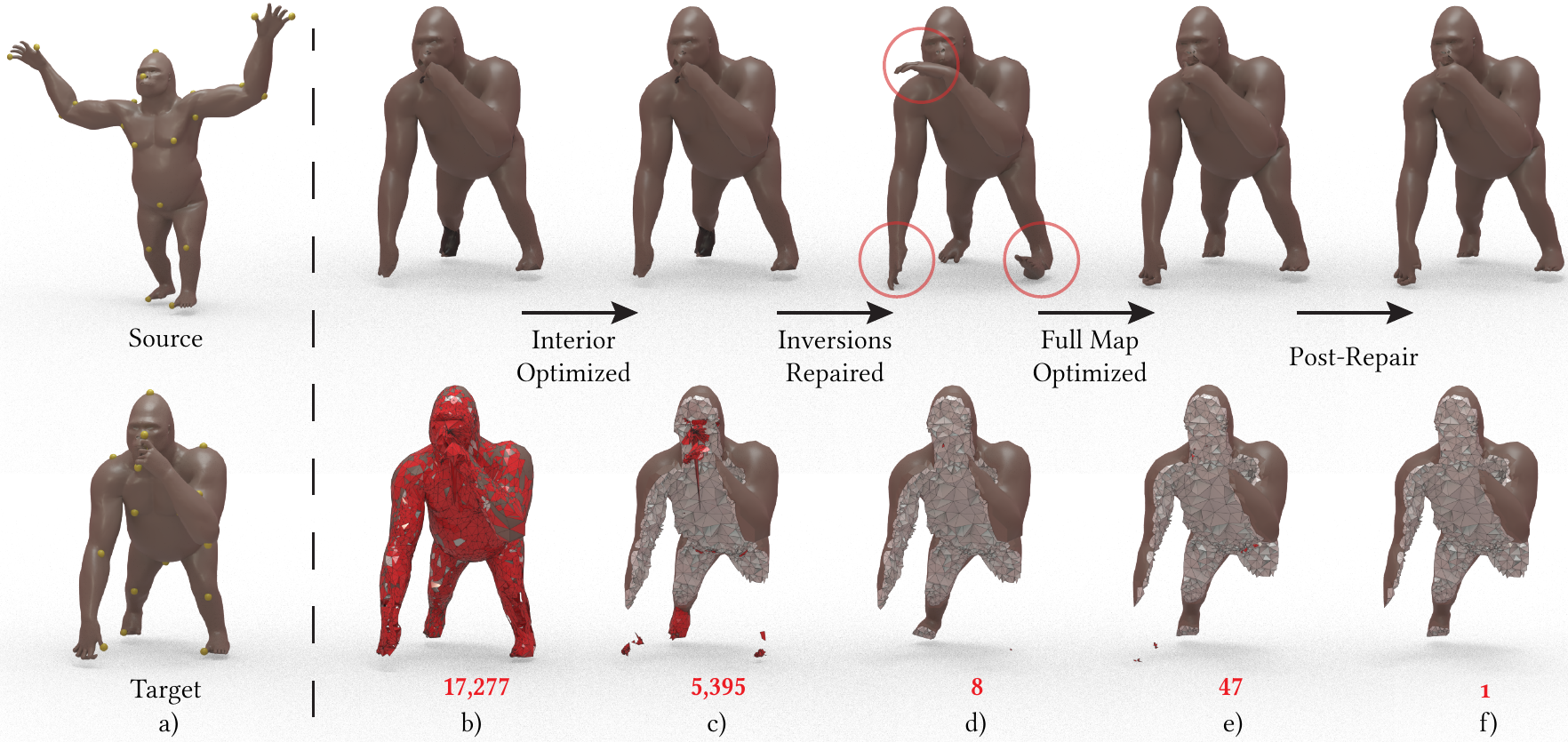}
    \caption{Flowchart depicting each step of our method: a) initial source and target shapes, with landmarks shown as yellow spheres; mapped shape; b) at initialization; c) after optimization converges while keeping the boundary fixed; d) after tetrahedron inversion repair; e) at convergence; \resub{and f) after post-convergence tetrahedron repair}. Top row shows the boundary of the mapped shape at every step and the bottom row shows a cut through the interior, revealing interior tetrahedra. Inverted and collapsed tetrahedra are red. The number of inverted tetrahedra is listed under each cut-through mesh. Our initial map b) has all interior tetrahedra collapsed to the boundary, resulting in 17,277 (46\%) degenerate or flipped tetrahedra. Steps c) and d) optimize and repair the interior, resulting in \resub{8} flipped tets. The tetrahedron repair step restores elements of the map to match the source, as the hands and feet rotate. The final optimization \resub{followed by the post-convergence repair} produces a map that closely matches the boundary with \resub{negligible inversions (1 flipped tetrahedron)}. }
    \label{fig:flowchart}
\end{figure*}

\subsection{Optimization Problem}

\resubb{Combining the distortion and regularization energies, our optimization problem becomes}
\begin{equation}
\label{eq:opt-problem-discrete}
\begin{aligned}
\argmin_{P_{12},P_{21},X_{12},X_{21}} \quad & {E}[P_{12},P_{21},X_{12},X_{21}]\\
\textrm{subject to} \quad & P_{12} \in \mathcal{P}_{12} \, , P_{21} \in \mathcal{P}_{21},\\
\end{aligned}
\end{equation}
where
\begin{equation}
\begin{aligned}
\label{eq:total-e-halfquad}
    E[P_{12},P_{21},X_{12},X_{21}] &= \\
    \sum_{\substack{i,j \in \{1,2\} \\  i\neq j }}\alpha E_{ARAP}[X_{ij}] &+ (1-\alpha)E_{R}[P_{ij},X_{ji}] \\&
    + \gamma E_P[X_{ij}] + \beta E_Q[X_{ij},P_{ij}].
    \end{aligned}
\end{equation}

Several parameters govern the strength of the distortion energies and soft constraints. The parameter $\alpha \in [0,1]$ models the tradeoff between a reversible map (small $\alpha \rightarrow 0$) and one that maintains the rest shape ( $\alpha \rightarrow 1$). The parameter $\gamma \in \R_{\geq 0}$ weighs the projection term that models the soft constraint for matching to the target boundary. The parameter $\beta$ controls the soft constraint on the auxiliary variables. As recommended by~\cite{ezuz2019reversible,wang2008new}, $\beta$ should use an update schedule tailored per application. In our experiments, since we start with a coarse initialization of the interior, we initialize $\beta = 0.25$ and increase $\beta$ linearly to $5$ over $20$ iterations. We found our approach to be insensitive to the update schedule. 

In this formulation, we use a soft constraint measured by $E_P$ to map to the target boundary. While we could use a hard constraint by setting $\gamma=0$ and requiring $P_{12} \in \mathcal{P}^{\star}_{12},P_{21} \in \mathcal{P}_{21}^{\star}$, we did not find that this hard constraint had a substantial effect on our final output.

\begin{figure}[t]
    \centering
    \includegraphics[width=0.9\linewidth]{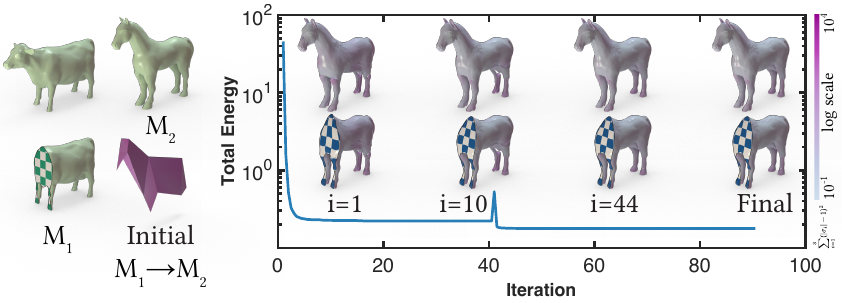}
    \caption{\resubb{Optimization of Eq.~\eqref{eq:opt-problem-discrete} using a landmark initialization. Despite a coarse initialization, our algorithm approximates the target shape after one iteration. Further optimization decreases surface distortion and improves interior regularity as visualized by the checkerboard patterns. At iteration $41$, the inverted tetrahedron repair is performed, causing a jump in the projection energy $E_P$, from which our algorithm quickly recovers.}}
    \label{fig:optimization-graph}
\end{figure}

\section{Optimization}

In this section, we outline our optimization procedure. We discuss strategies for initializing the map and propose an approach to uninvert tetrahedra. We conclude by presenting our algorithm for minimizing Eq.~\eqref{eq:opt-problem-discrete} using block coordinate descent.

\subsection{Initialization}\label{sec:initialization}

Objective function~\eqref{eq:opt-problem-discrete} includes four variables:  $P_{12},$ $P_{21},$ $X_{12},$ and $X_{21}$.  In this section, we provide strategies for initializing the variables $P_{ij}$ before running our optimization procedure. We initialize the $X_{ij}$ variables via $X_{ij} \gets P_{ij}V_j$. 

\paragraph*{Landmark-based initialization.} If we are given landmark pairs $(\mathbf{p}_{i},\mathbf{q}_{i})$, where $\mathbf{p_{i}} \in M_1, \mathbf{q_{i}} \in M_2$, we can initialize each landmark's target by copying the target of its closest landmark. 


\paragraph*{2D surface map initialization.} A second approach is to initialize the boundaries of $M_1,M_2$ using an existing surface-to-surface mapping approach. We initialize the interior vertices identically to landmark-based initialization, where we consider \emph{every} boundary vertex to be a landmark.

\resubb{We do not hold the landmark or surface map vertices fixed during the optimization.}




\subsection{Alternating Minimization}\label{sec:alternating}
We use coordinate descent, alternating between optimizing over $X_{ij}$ and $P_{ij}$. Our multi-step optimization procedure ensures strong conformation to the boundary while avoiding inverted tetrahedra. 

\paragraph*{Optimizing for $X_{ij}$} Optimizing for $X_{ij}$ while holding the $P_{ij}$ variables fixed is a smooth optimization problem, for which we use the Limited-memory Broyden–Fletcher–Goldfarb–Shanno (L-BFGS) algorithm~\cite{zhu1997algorithm}.

We compute the gradient of each energy term in Eq.~\eqref{eq:total-e-halfquad}. The gradients for $E_P, E_Q$ are straightforward as they are matrix norms. We compute the gradient of $E_{ARAP}$ using the chain rule. First, we compute the gradient of $f_{ARAP}(J)$ with respect to a Jacobian  $J$,
$\nabla_J f_{ARAP}[J]=U\mathrm{diag}\left(\nabla_{\sigma}f_{ARAP}\left(\sigma\right)\right)V^T.$ Using the chain rule, we then compute the gradient with respect to the elements of tetrahedron $T_k \in \mathcal{T}_i$, with coordinates $X_{ij}^{T_k}$,
\begin{equation}
\frac{\partial f_{ARAP}(X_{ij}^{T_k})}{\partial (X_{ij}^{T_k})} = \left(\left(BV_i^{T_k} \right)^{-T} B \right)\left(U\mathrm{diag}\left(\nabla_{\sigma}f_{ARAP}\left(\sigma\right)\right)V^T\right)^T
\end{equation}

The gradient with respect to each vertex is found by gathering the gradients of each tetrahedron adjacent to that vertex.

\paragraph*{Optimizing for $P_{ij}$} Fixing $X_{12},X_{21}$, the remaining energy terms with respect to $P_{ij}$ are of the form $\|P_{ij}A-B\|_{M_i}^2$ with $A \in \R^{n_j \times 6},B \in \R^{n_i \times 6}$. Following~\citet{ezuz2019reversible}, this minimization can be understood as a \emph{projection} problem solved independently for each row of $P_{ij}$. 

In our case, we need to project the points in $A$ to the $6-$dimensional tetrahedral mesh with vertices $B$, whose connectivity is the same as $M_j$. The presence of several additional energy terms in our formulation also leads to a unique projection problem. Since the problem can be solved independently, we implement an efficient solution using CUDA programming.
To enforce a hard boundary-to-boundary constraint, we map rows of $A$ corresponding to the boundary of $M_i$ to the boundary of the target embedding. 


\subsection{Inverted Tetrahedron Repair}\label{sec:repair}

The initial maps suggested in \S\ref{sec:initialization} are straightforward to compute, but they are quite distant from our desired output; indeed, the majority of tetrahedra in our initial maps have zero volume. Although alternating between the two steps above is guaranteed to decrease the objective function in each step, empirically we find in the initial stages our algorithm can get stuck in local optima due to inverted elements.  Here, we describe a heuristic strategy that empirically can improve the quality of our output.

In this tetrahedron repair step, we find all inverted tetrahedra. We then take the $1-$ring neighborhood of the vertices in the inverted tetrahedra and use L-BFGS to minimize  $f_{ARAP}$ with the remaining vertices fixed. 
%



\subsection{Full Algorithm and Stopping Criteria} 

Overall, our optimization procedure follows four broad steps:
\begin{enumerate}
    \item map initialization (\S\ref{sec:initialization});
    \item optimization while keeping the boundary fixed (\S\ref{sec:alternating});\label{stage2}
    \item inverted tetrahedron repair (\S\ref{sec:repair});
    \item optimization of all vertices (\S\ref{sec:alternating}); and\label{stage4}
    \item \resub{post-convergence inverted tetrahedron repair (\S\ref{sec:repair})}\label{stage5}.
\end{enumerate}

For stages \ref{stage2} and \ref{stage4}, we set as our convergence criteria one of (i) the norm of the gradient $<10^{-6}$, (ii) the objective function decreases by less than $10^{-7}$ between successive iterations, or (iii) run for $50$ iterations; the third criterion is a fallback that rarely occurs in practice. \resub{For stage \ref{stage5}, we limit vertex displacement to preserve map quality by limiting to 100 steps of L-BFGS} \resubb{and we restrict optimization to only vertices in inverted tetrahedra. }

Algorithm \ref{alg:mapping} summarizes our full procedure.

\begin{algorithm}[t]
\caption{Coordinate decent with tetrahedra uninversion}\label{alg:mapping}
\begin{flushleft}
        \textbf{Input:} initial maps  $P_{12},P_{21}$\\
        \textbf{Output:} \resub{optimized maps $X_{12}$, $X_{21}$, $P_{12}$, $P_{21}$ }
\end{flushleft}
\begin{algorithmic}[1]
\State $\partial P_{12}^{(0)} \gets P_{12}(\partial V_1,:)$ \textit{// initial boundary map}
\State $\partial P_{21}^{(0)} \gets P_{21}(\partial V_2,:)$
\State $X_{12}\gets P_{12}V_2$ \textit{// initial vertex map}
\State $X_{21}\gets P_{21}V_1$
\\
\While{!converged} \textit{// optimize boundary map}
    \For{$(i,j)\in\{ (1,2), (2,1) \}$}
        \State $P_{ij} \gets \argmin_{P \in \mathcal{P}_{ij}} \bar{E}_R[P,X_{ji}] + \bar{E}_Q[P,X_{ij}]$
        \State $X_{ij} \gets \argmin_{X \in \R^{n_i \times 6}} \bar{E}_{ARAP}[X_{ij}]$ \par \hskip\algorithmicindent \hskip\algorithmicindent \hskip\algorithmicindent $+ \bar{E}_R[X_{ij},P_{ji}] + \bar{E}_P[X_{ij}] + \bar{E}_Q[X_{ij},P_{ij}]$
        \State $\partial P_{ij} \gets \partial P_{ij}^{(0)}$ \textit{// restore boundary}
    \EndFor
\EndWhile
\\\ \\
\textit{// inverted tetrahedron repair}
\State $\textrm{idx} \gets \det J(X_i^{T_k})\leq 0, \forall T_k \in \mathcal{T}_i$ \textit{// find inverted tetrahedra}
\State $X_{ij}(idx) \gets \argmin_{X \in \R^{n_i \times 6}} \bar{E}_{ARAP}[X_{ij}(idx)]$ \textit{// 1-ring nbhd.}
\\
\While{!converged} \textit{// optimize full map}
    \For{$(i,j)\in\{ (1,2), (2,1) \}$}
        \State $P_{ij} \gets \argmin_{P \in \mathcal{P}_{ij}} \bar{E}_R[P,X_{ji}] + \bar{E}_Q[P,X_{ij}]$
        \State $X_{ij} \gets \argmin_{X \in \R^{n_i \times 6}} \bar{E}_{ARAP}[X_{ij}]$ \par \hskip\algorithmicindent \hskip\algorithmicindent \hskip\algorithmicindent $+ \bar{E}_R[X_{ij},P_{ji}] + \bar{E}_P[X_{ij}] + \bar{E}_Q[X_{ij},P_{ij}]$
    \EndFor
\EndWhile
\end{algorithmic}
\end{algorithm}

\subsection{Implementation Details}

Unless otherwise noted, all figures are generated using identical parameters.  
We use grid search to identify reasonable parameters; the results of our analysis are provided in Fig.~\ref{fig:paramsweep}. We set the rigidity parameter $\alpha=0.5$ and the boundary conformation parameter $\gamma=25$, achieving a reasonable trade off between average distance to the target \resub{and maintaining per-tetrahedron map quality as measured using $\det\hat{J}$, the normalized Jacobian determinant}. To find these values, we initialize $\beta=0.25$ and increase linearly to $\beta=5$ over 20 iterations. In practice, we found our method was insensitive to the choice of $\beta$.

We generate tetrahedral meshes using fTetWild~\cite{hu2020ftetwild}. Prior to mapping, we normalize each mesh to have volume 1.  We perform one tetrahedron repair step as we found negligible improvement after performing more. 

We implement our method in MATLAB, using CUDA to optimize the projection step by extending the projection code in~\cite{li2021interactive} to $\R^6$. Our code and data are available at \url{https://github.com/mabulnaga/symmetric-volume-maps}.
\section{Experiments}
We measure map quality by assessing distortion and closeness to matching the target shapes (\S\ref{sec:qualitymetrics}). We validate our method by mapping pairs of shapes from four datasets (\S\ref{sec:datasets}) and report visualizations and numerical scores evaluating the result (\S\ref{sec:validation}). \resub{We also compare our method to several variants of a baseline mapping approach (\S\ref{sec:baseline-comparison}).} \resub{We test the robustness of our method in \S\ref{sec:robust}} and evaluate the choice of symmetrized energy on computing a map in \resub{\S\ref{sec:symmetrized_energy_choice}.}

\begin{figure}
    \includegraphics[width=\linewidth]{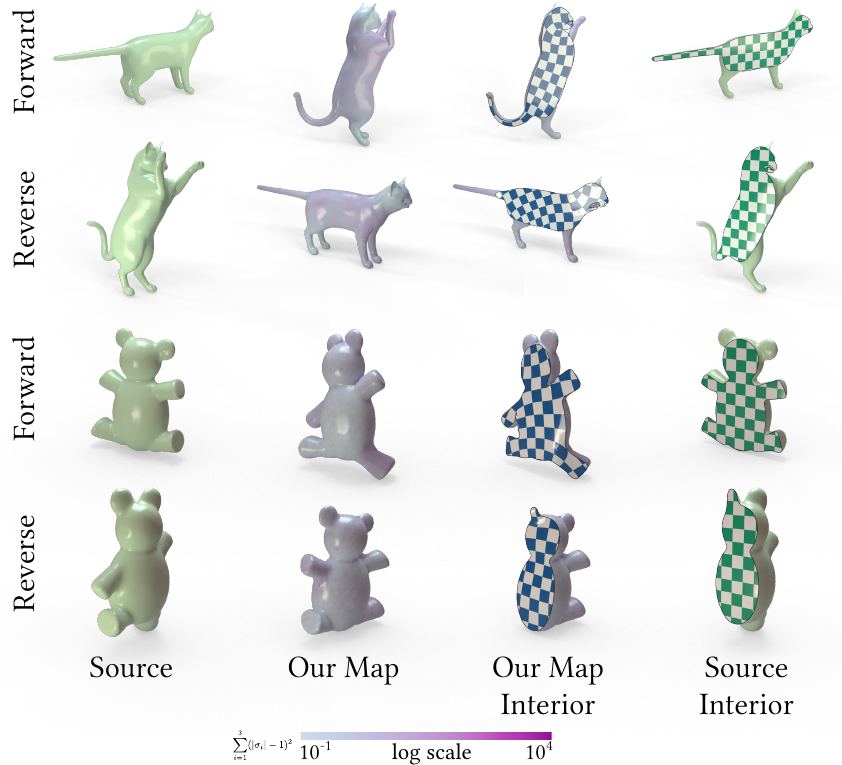}
    \caption{Forward and reverse maps on related pairs of shapes. We observe smooth patterns of distortion on the boundary while capturing distinguishing geometric features, such as the \resubb{transformation} of the tail of the cat and movement of the bear's ears. Distortion is uniform throughout the interior.}
    \label{fig:checkeriso}
\end{figure}

\begin{table}[bp]
\scriptsize
  \caption{Map Quality Evaluation }\vspace{-.15in}
  \label{tab:map-comparison}
  \begin{tabular}[0.3\textwidth]{|lccccccc|}
    \hline
    \makecell{Map\\ (Initialization)} & \makecell{Time \\ (min.)} & \makecell{$E_R$\\ $(\times 10^{-3})$} & \makecell{$E_{ARAP}$\\ $(\times 10^{-3})$} &$n_{inv}$ & \makecell{$\hat{d}_{max}$\\($\times 10^{-2}$)} & \makecell{$\hat{d}_{avg}$\\($\times 10^{-2}$)} &$\det\hat{J}$\\
    \hline
    \makecell{$X_{ij}$ \\ (Surface)} & \makecell{$31$ \\ $\pm21$} & \makecell{$1.47$ \\ $\pm 1.9$} & \makecell{$81.7$ \\ $\pm 78.5$}& \makecell{$7.7$ \\ $\pm9.1$} &\makecell{$2.5$ \\ $\pm1.2$} &\makecell{$0.10$ \\ $\pm0.046$} &\makecell{$0.98$ \\ $\pm0.02$} \\ \hline
    
   \makecell{$P_{ij}$ \\ (Surface)} & \makecell{$31$ \\ $\pm21$} &  \makecell{$1.29$ \\ $\pm 1.65$} &  \makecell{$134.5$ \\ $\pm 115.4$} & \makecell{$649$ \\ $\pm549$} &\makecell{$1.9$ \\ $\pm0.78$} &\makecell{$0.072$ \\ $\pm0.028$} &\makecell{$0.96$\\ $\pm0.04$} \\ \hline
   
  \makecell{$X_{ij}$ \\ (Landmark)} & \makecell{$107$ \\ $\pm 53$} & \makecell{$7.45$ \\ $\pm 10.7$} &  \makecell{$93.6$ \\ $\pm 73.3$} & \makecell{$15.8$ \\ $\pm10.9$} &\makecell{$2.7$ \\ $\pm1.0$} &\makecell{$0.12$ \\ $\pm0.046$} &\makecell{$0.97$\\ $\pm0.02$}\\ \hline
  
  \makecell{$P_{ij}$ \\ (Landmark)} & \makecell{$107$ \\ $\pm 53$} & \makecell{$6.67$ \\ $\pm 9.7$} &  \makecell{$176.6$ \\ $\pm 145.4$} & \makecell{$723$ \\ $\pm515$} &\makecell{$2.6$ \\ $\pm1.0$} &\makecell{$0.11$ \\ $\pm0.038$} &\makecell{$0.94$\\ $\pm0.04$}\\ 
    \hline
  \end{tabular}
\end{table}

\subsection{Quality Metrics}\label{sec:qualitymetrics}

We validate our method using the metrics outlined below. 

\paragraph*{Boundary matching.} We measure fit to the target boundary using the Hausdorff distance $d_{\textrm{max}}$ and the chamfer distance $d_{\textrm{avg}}$ defined as follows:
\begin{align}
    d_{\textrm{max}} (M_1,M_2) &= \max\left\{\sup_{\mathbf x \in M_1} \inf_{\mathbf y \in M_2} d(\mathbf x,\mathbf y), \sup_{\mathbf y \in M_2} \inf_{\mathbf x \in M_1} d(\mathbf x,\mathbf y) \right\}
    \\
    d_{\textrm{avg}} (M_1,M_2) &= \frac{1} {\vert \mathcal{V}_1 \vert + \vert \mathcal{V}_2 \vert}\left[
    \sum_{\mathbf v_i \in \mathcal{V}_1} d(\mathbf v_i,M_2) + \sum_{\mathbf v_j \in \mathcal{V}_2}d(\mathbf v_j,M_1)
    \right].
\end{align}
Here, $\mathcal V_1$ and $\mathcal V_2$ denote the sets of vertices of $M_1$ and $M_2$, respectively. 
To make the measures above scale-independent, we normalize both quantities by the length of the diagonal of the bounding box enclosing the target mesh.  We use hats to denote normalized quantities:  $\hat{d}_{\max}$ and $\hat{d}_{\mathrm{avg}}$.

To visualize the distortion in the interiors of tetrahedral meshes, we use a mapped checkerboard pattern. In each map visualization, using Houdini, we slice the source shape with a plane and place an extrinsic checkerboard pattern on the intersection, using rounding and modulo operations on coordinates. We push forward the planar intersection surface through our map and render the result using a custom shader that looks back to the corresponding coordinate in the source and evaluates the checkerboard function.
Interpolation happens by finding the closest element (\texttt{xyzdist}) and then transferring coordinates (\texttt{primuv}).

\begin{figure}
    \includegraphics[width=\linewidth]{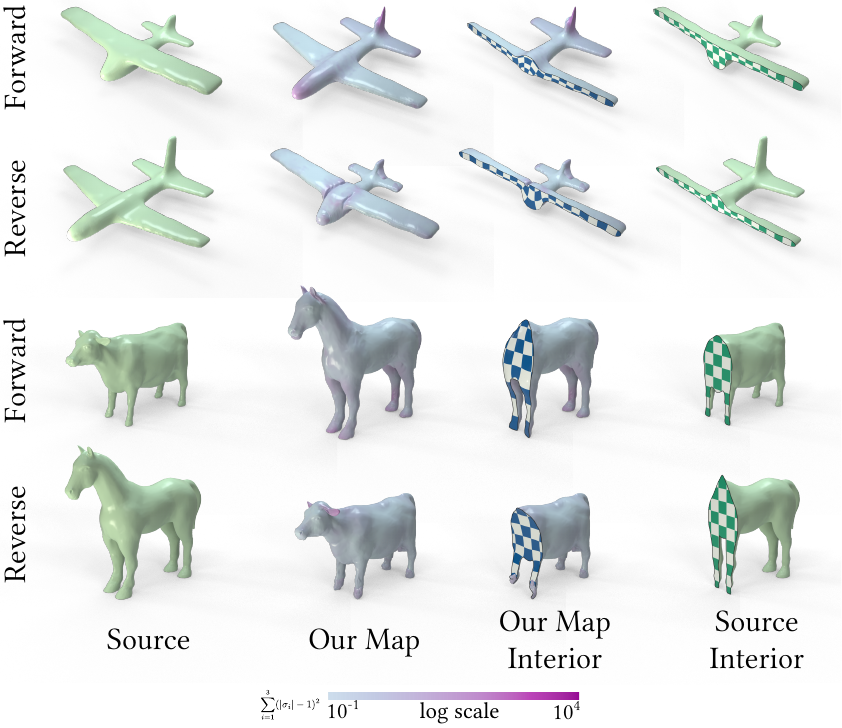}
    \caption{Forward and reverse maps on far-from-isometric shapes. Our maps capture the extreme deformations, for example by growing and collapsing the airplane rudder and  deforming the ears of the horse and cow pair. \resubb{Matching boundary features expectantly leads to high local distortion, as a large volume change is required to model these transformations. The checkerboard pattern reveals that regions with high boundary distortions also cause interior distortion (see airplane), but the computed maps are uniform and smooth elsewhere.}}
    \label{fig:checkernoniso}
\end{figure}

\paragraph{Distortion and inversion.} 
\resub{We measure the quality of the transformation by computing the number of inverted tetrahedra $(n_{inv})$} and the mean normalized Jacobian determinant $\det \hat{J}$ (weighted by tetrahedron volume), where the columns of $J$ are normalized as in~\cite{li2021interactive}. Figures containing qualitative results depict distortion per tetrahedron using the ARAP energy $\sum_{i=1}^3 \left( \vert \sigma_i \vert -1 \right)^2$.

\subsection{Datasets}\label{sec:datasets}
We evaluate our method on 24 mesh pairs from four datasets. For datasets where only triangle meshes are available, we tessellate the interiors. We randomly select pairs of shapes distorted non-isometrically from the SHREC19 dataset~\cite{dyke2019shrec}. We also randomly select matching and non-matching pairs of humans and animals for nonrigid correspondence from the TOSCA dataset~\cite{bronstein2008numerical}. Finally, we obtain tetrahedral meshes of models of natural objects and CAD models from~\cite{li2021interactive,Fu-2016-PC}, from Thingi10k~\cite{Thingi10K}, and from Thingiverse~\cite{thingi2022yahoo}. \resub{The resulting meshes had (mean$\pm$standard deviation) $50,010\pm34,663$ tetrahedra.}  We manually choose landmarks on the boundary surfaces for every mapping example (marked on most figures); \resub{Table~\ref{tab:all-stats-rhm-init} provides the number of landmarks and number of tetrahedra for each pair}.

\subsection{Validation}\label{sec:validation}
In this section, we demonstrate our maps on several pairs. 

\begin{figure}
    \includegraphics[width=\linewidth]{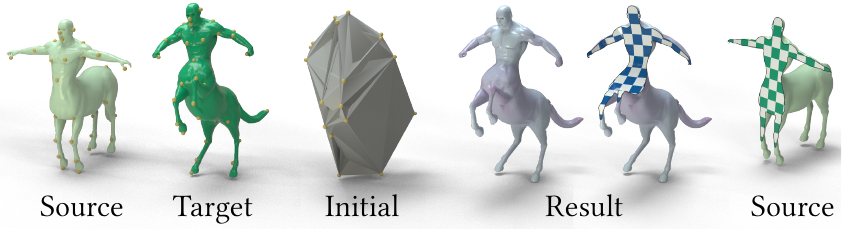}
    \caption{Resulting map when initialized using only a sparse set of landmark points. Despite an initialization that collapses the mesh to a set of landmarks, we produce a map that captures sharp geometric features of the target including the hands and bends of the legs. The distortion is smooth and uniform throughout the boundary and interior.}
    \label{fig:landmark}
\end{figure}

\resub{\paragraph*{Quantitative evaluation and map selection}Table~\ref{tab:map-comparison} measures performance of both sets of maps, $X_{ij}$ and $P_{ij}$, using surface map initialization and landmark initialization. Using the image of the map $X_{ij}$, we achieve close matchings to the target boundary with negligible tetrahedron inversions and while effectively maintaining tetrahedron quality. The landmark-based initialization achieves comparable performance, with slightly higher $\hat{d}_{max}$. These results indicate our method is robust to the choice of initialization. The constrained maps $P_{ij}$ have significantly higher tetrahedron inversion due to the constraint $P_{ij}\in\mathcal{P}_{ij}$, which results in boundary tetrahedron foldovers. Since the boundary matching metrics are comparable for both maps, we use $X_{ij}$ as the final map. The low number of tetrahedron inversions ($n_{inv}$) and small $E_R$ indicate the resultant maps are nearly inverses of one another.  Table~\ref{tab:all-stats-rhm-init} presents results for all pairs in our dataset.
}

\paragraph*{Algorithm flowchart} We demonstrate each step of our algorithm in Fig.~\ref{fig:flowchart}. First, we compute an initial boundary map using the method of \citet{ezuz2019reversible}. This initial map is interpolated from the boundary to the interior by mapping each interior vertex to the target of its closest boundary vertex, as described in \S\ref{sec:initialization}. This procedure results in a significant number of inverted or collapsed tetrahedra (Fig.~\ref{fig:flowchart}b). The interior is then improved by minimizing the map energy over the interior vertices (Fig.~\ref{fig:flowchart}c). Then, we repair inverted tetrahedra, dramatically reducing the number of flipped tetrahedra, as described in \S\ref{sec:repair}. The mapped mesh start to restore its source pose; the hands and feet rotate (Fig.~\ref{fig:flowchart}d). \resub{We compute the final map by optimizing over all vertices (Fig.~\ref{fig:flowchart}e) and then perform post-convergence tetrahedron repair, arriving at a solution that closely conforms to the target boundary while minimizing distortion (Fig.~\ref{fig:flowchart}f)}.

\resubb{Fig.~\ref{fig:optimization-graph} visualizes our optimization routine initialization with landmarks. A few intermediate shapes are demonstrated. Our algorithm quickly recovers the target shape and the optimization improves surface matching, and reduces boundary and interior distortion.}

\paragraph*{Map results} We demonstrate our method on several pairs.  Fig.~\ref{fig:checkeriso} shows the forward and reverse maps between pairs of deformations from the same domain. In both examples, distortion is smooth throughout the boundary, and our map successfully matches geometric features, for example the curved tail and the ears in the cat pairs. The checkerboard patterns demonstrate that our maps are smooth in the interior. 

\begin{figure}
    \includegraphics[width=\linewidth]{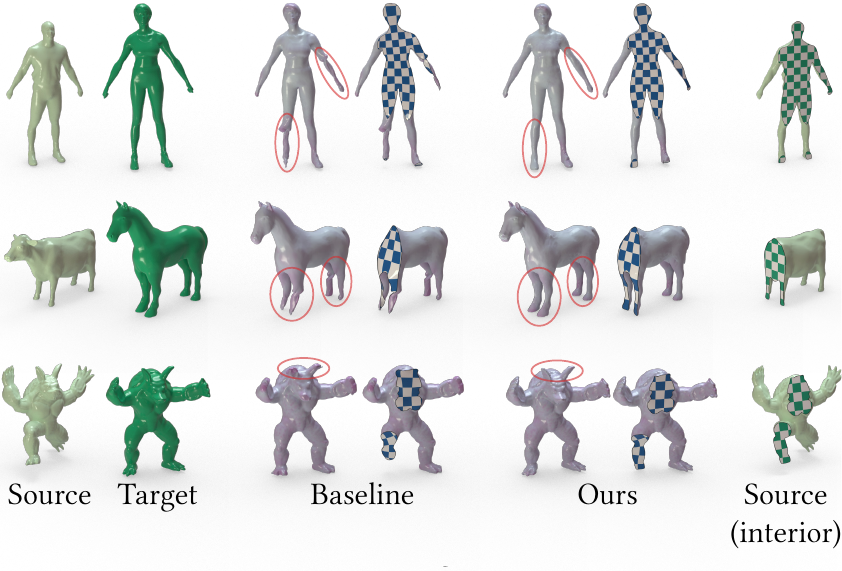}
    \caption{\resub{Comparison of our map with the baseline approach using $K=25$ with landmark equality constraints. Red ovals indicate distorted regions in the baseline where our method succeeds. Our approach effectively preserves geometric features and produces high quality maps.}}
    \label{fig:baseline}
\end{figure}

Fig.~\ref{fig:checkernoniso} shows results for the more challenging problem of mapping between pairs of shapes from different domains. \resubb{Distortions are mainly smooth on the boundary but are expectantly high in regions with large displacements, e.g., in the nose and rudder of the airplane in the forward direction. }Here, the volume of the nose has to shrink substantially while the rudder has to expand in height. Similarly, we see large distortion in the cow-horse pair, particularly in the ears in the reverse map and in the knees and feet in the forward map. \resub{Our boundary term yields maps that closely conform to the target at the cost of greater tetrahedral distortion}.

\resub{Fig.~\ref{fig:landmark} demonstrates our resultant map when initialized using a sparse set of landmark points (\S\ref{sec:initialization}, landmark-based initialization). While the initial map is unintelligible, our output matches the target shape closely. The final map has low distortion throughout the boundary and captures the narrow features of the target, including the fingers and bends in the legs. Furthermore, the checkerboard pattern reveals uniform distortion in the interior. }

\resub{
\subsection{Baseline Comparison}\label{sec:baseline-comparison}
We compare to the volumetric mapping approach of~\citet{kovalsky2015large}. Their method inputs a surface map with optimized interior and computes a similar map that is orientation-preserving with bounded condition number $K$. Linear equality constraints on the vertices are used to fix parts of the map. 

We compute the initial volumetric map by first computing a surface map as in \S\ref{sec:initialization} and then repairing degenerate tetrahedra by minimizing the Dirichlet energy while keeping the boundary fixed, as was done by~\citet{kovalsky2015large}. We test four different sets of equality constraints for extracting the final volumetric maps: (1) fixing the boundary map; (2) fixing the boundary map for vertices not in inverted tetrahedra; (3) fixing landmarks; and (4) preserving center of mass. We use conformality bound $K \in \{5,25,50,100\}$.}

\begin{table}[bp]
\scriptsize
  \caption{\resub{Map quality comparison to the baseline for $K=25$.} \vspace{-.15in}}
  \label{tab:map-comparison-baseline-k}
  \begin{tabular}[0.3\textwidth]{|clcccc|}
    \hline
        & Constraint & $n_{inv}$ & \makecell{$\hat{d}_{max}$\\($\times 10^{-2}$)} & \makecell{$\hat{d}_{avg}$\\($\times 10^{-2}$)} &$\det\hat{J}$  \\ \hline
        & \textbf{Ours} & $8\pm13.8$ & $\mathbf{2.35\pm1.45}$ & $0.097\pm0.05$ & $\mathbf{0.98\pm0.02}$ \\ \hdashline[1pt/1pt]
        \parbox[t]{1mm}{\multirow{4}{*}{\rotatebox[origin=l]{90}{\hspace{-.05in}\textbf Baseline}}} & Boundary & $2740\pm2210$ & $2.84\pm1.06$ & $\mathbf{0.085\pm0.049}$ & $0.82\pm0.17$ \\
        &Boundary (no flip) & $11.1\pm31.8$ & $7.9\pm8.5$ & $0.33\pm0.38$ & $0.89\pm0.1$ \\ 
        &Landmark & $1.8\pm3.2$ & $5.2\pm3.1$ & $0.7\pm0.26$ & $0.89\pm0.11$ \\ 
        &Center of mass & $\mathbf{1.7\pm2.6}$ & $7.2\pm0.58$ & $0.8\pm1.2$ & $0.89\pm0.11$ \\ \hline
    \end{tabular}
\end{table}

\begin{figure}
    \includegraphics[width=\linewidth]{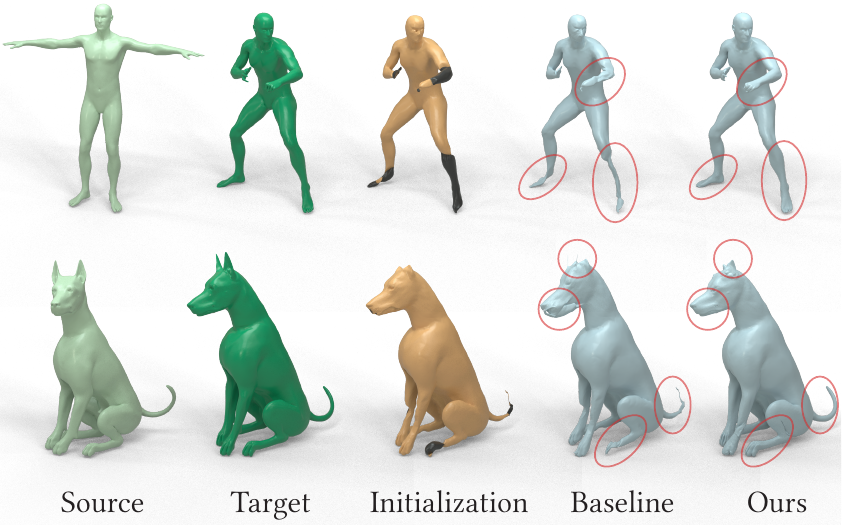}
    \caption{Refinement of the initial boundary map using~\cite{ezuz2019reversible} \resub{and comparison to the baseline with landmark equality constraints}. The backs of boundary triangles are shown in black. The initial map produces areas of the surface turned inside out and collapses regions like the hands of the human and tail of the dog. \resub{Both our method and the baseline can produce orientation-preserving correspondences}. \resub{Compared to the baseline, our approach restores collapsed and distorted regions and effectively matches the target shape (red ovals)}. This experiment also reveals that our method can recover from poor initialization. }
    \label{fig:danielle-comp}
\end{figure}

 \resub{Table~\ref{tab:map-comparison-baseline-k} compares map quality across the dataset for each equality constraint using $K=25$. Similar behavior arose for other values of $K$, so they are not shown. We compare with the matching forward maps from our method. The fixed boundary map results in comparably low $\hat{d}_{\textrm{max}}$, $\hat{d}_{\textrm{avg}}$ to our method, but with a significantly large number of flipped tetrahedra and poor map quality ($\det \hat{J}=0.82$) compared to our approach ($\det \hat{J}=0.98$). The strongest baseline uses the landmark equality constraints, resulting in improved $n_{inv}$, at the cost of map quality and boundary matching. 
 
 Fig.~\ref{fig:baseline} compares our map with the baseline using the fixed landmark constraint. Our method correctly maps features that are distorted by the baseline, such as the arm and leg of the human and hooves of the horse. The baseline approach performs well on the armadillo, a map between shapes of the same domain, but produces higher distortion. These visual and quantitative results demonstrate the strength in our free-boundary formulation, which effectively matches geometric features.}


 \begin{figure}
    \includegraphics[width=\linewidth]{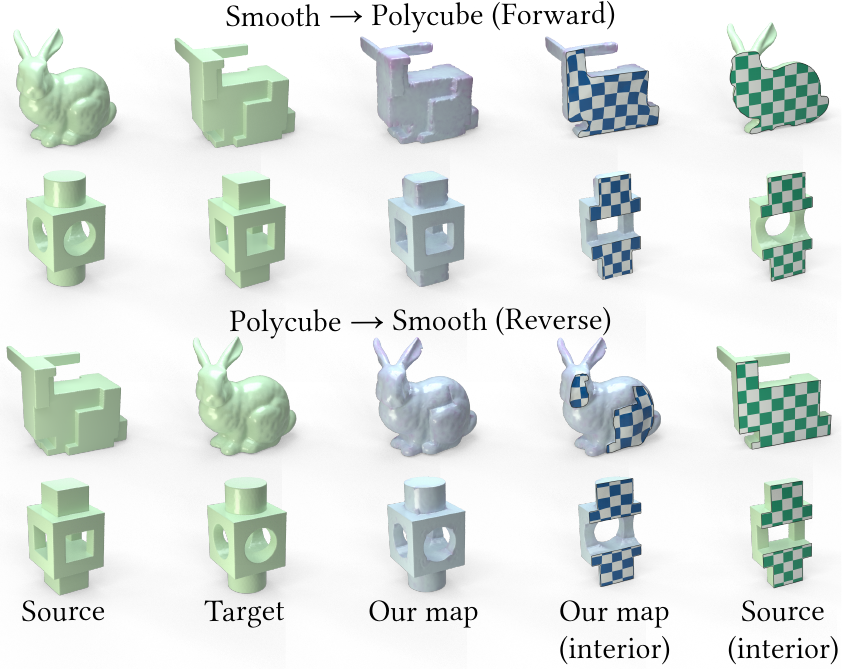}
    \caption{\resub{Map between smooth and polycube shapes. Our method produces close matchings in both directions, though higher distortion arises in the corner regions of the polycube. }}
    \label{fig:polycube}
\end{figure}

\paragraph*{Surface map repair} 
Fig.~\ref{fig:danielle-comp} shows how our algorithm recovers artifacts in the 2D surface map initialization procedure (\S\ref{sec:initialization}) and compares with the baseline using the fixed landmark constraint. 
Starting from our landmarks, \cite{ezuz2019reversible} results in parts of the surface that are folded inside out (the backs of triangles are shown in black), as seen on the arms and legs of the human and the paws of the dog; the initial maps also have collapsed boundary features (hand of the human, tail of the dog). \resub{ Both our method and the baseline target orientation-preserving maps and correct these inverted areas. Unlike the baseline, our method recovers from the inverted regions to match the target shape. Furthermore, we fill small regions such as the tail of the dog and the hands and feet of the human. }

\begin{figure}[b]
    \includegraphics[width=\linewidth]{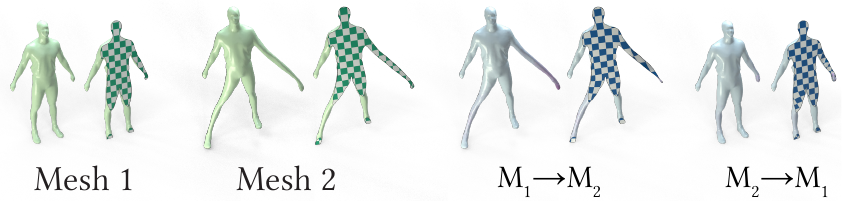}
    \caption{\resub{Nonisometric mapping of a stretched human. Close matchings are obtained, though higher distortion arises in the stretched regions. }}
    \label{fig:human-stretch}
\end{figure}

\begin{figure}[t]
    \includegraphics[width=0.9\linewidth]{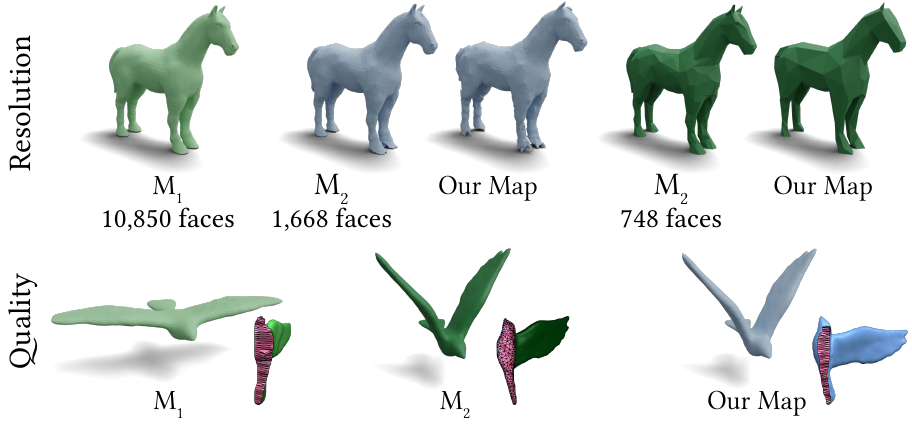}
    \caption{\resubb{Map sensitivity to mesh quality. Top: mapping a high-resolution horse mesh to progressively downsampled versions (boundary triangle faces indicated). Bottom: mapping a bird with thin, elongated tetrahedra faces to one with regular tetrahedra. In all cases, the targets are matched closely with few inversions (maximum of $n_{inv}=2$), though in the horse examples, small geometric features, such as the ears, are lost due to limited representation. }}
    \label{fig:coarse}
\end{figure}

\resub{
\subsection{Map Robustness}
\label{sec:robust}
We test the robustness of our method on challenging mapping cases. We first assess the ability to map from smooth, high-resolution shapes to coarse meshes with sharp features. Fig.~\ref{fig:polycube} demonstrates mapping to polycube shapes from~\cite{Fu-2016-PC}, using the $P_{ij}$ maps. We successfully map bidirectionally between the smooth and coarse shapes, although expectantly higher distortion arises in the corner regions. 

Fig.~\ref{fig:human-stretch} tests matching between nonisometric pairings. \resub{We stretch one arm and leg of the human mesh and obtain close matchings in both directions, although higher distortion arises at the ends of the stretched regions due to large changes in volume required to match to the target.}

\resubb{Fig.~\ref{fig:coarse} tests the robustness of our method to mesh quality. Fig.~\ref{fig:coarse} (top)  maps a high-resolution horse to progressively downsampled versions. Despite differences in mesh resolution, we successfully map to the target shapes \resubb{with minimal inversions}, although small features like the ears of the horse are distorted. This artifact is due to few tetrahedra representing these regions in the downsampled mesh. Fig.~\ref{fig:coarse} (bottom) assess the sensitivity of our method to mesh quality by mapping a bird with thin, elongated tetrahedra faces to one with regular tetrahedra. We achieve a close matching, suggesting our method is robust to mesh quality.  }

}

\begin{figure}[bp]
    \includegraphics[width=\linewidth]{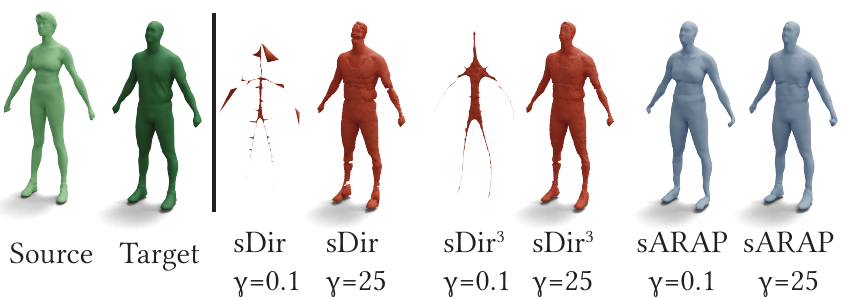}
    \caption{Comparison of maps when optimizing with the symmetrized Dirichlet (sDir), \resub{the \nth{3}-order sDir$^3$}, and the sARAP energies.  sDir and \resub{sDir$^3$} produce collapsed maps for both values of $\gamma$, although $\gamma=25$ keeps parts of the map intact as it pushes vertices to the boundary. The sARAP energy does not collapse, but starts to show the source shape for $\gamma=0.1$, as expected.}
    \label{fig:arapanddirichlet}
\end{figure}

\subsection{Symmetrized Energy Choice}\label{sec:symmetrized_energy_choice}
We experiment with the choice of symmetrized energy and its effect on producing a map. \resub{As described in \S\ref{sec:symm-design}, several symmetrized energies do not favor isometry while our choice, the sARAP energy, does. Fig.~\ref{fig:arapanddirichlet} compares the output when optimizing using the sARAP, the symmetrized Dirichlet (sDir), and the \nth{3}-order symmetrized Dirichlet (sDir$^3$) energies. The \nth{3}-order Dirichlet is used since tri-harmonic functions are used to achieve $C^1$ continuity in 3D~\cite{iwaniec2010deformations}}. In these experiments, we remove the tetrahedron repair step, which made the artifacts \resub{worse}. We compare two choices of $\gamma$ and visualize the resultant maps. 

Both the sDir and \resub{sDir$^3$} energy completely collapse the map for $\gamma=0.1$, since the projection term has little effect at keeping the map intact. Similarly, parts of the mapped mesh degenerate with $\gamma=25$. \resub{In both cases, the sDir$^3$ energy however maintains continuity}. In contrast, the sARAP energy does not produce a collapsed map, although it starts to restore the source when $\gamma=0.1$.

This experiment verifies our analysis in \S\ref{s:symm-energy} and additionally shows that methods using energies that do not favor isometry can be sensitive to parameter choice.

\section{Examples}

Volumetric maps are useful for transporting data between domains. Below, we depict some use cases that would benefit from our low-distortion, near-diffeomorphic maps.

\begin{figure}
    \includegraphics[width=\linewidth]{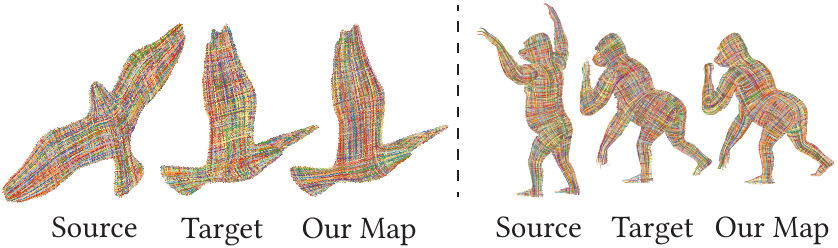}
    \caption{When the integral curves of an octahedral frame field are pushed forward from a source domain (left) to a target domain (right), the result looks similar to the integral curves of a field computed directly on the target (center). The mapped curves remain nearly orthogonal, illustrating the low metric distortion of our map.}
    \label{fig:singularitycurves}
\end{figure}

\subsection{Internal geometry transfer}

In contrast to pulling back functions on $M_2$ to $M_1$, we can also push forward maps into $M_1$ to $M_2$. This category of data includes point clouds, collections of curves, and arbitrary subdomains $U \subset M_1$.

As an example of how data can be easily transported using our maps, in \figref{singularitycurves} we push forward integral curves of a frame field on domain $M_1$ through $\phi: M_1 \to M_2$. The frame fields and their integral curves were generated using ARFF \cite{palmer2020}. 
Integral curves were pushed forward by mapping the curve vertices individually using piecewise linearity. The integral curves remain nearly orthogonal under the map, showing that it is close to isometric.
\begingroup
\setlength{\intextsep}{1pt}%
\setlength{\columnsep}{6pt}%
\begin{wrapfigure}{r}{0.18\textwidth}
  \centering
    \includegraphics[width=0.17\textwidth]{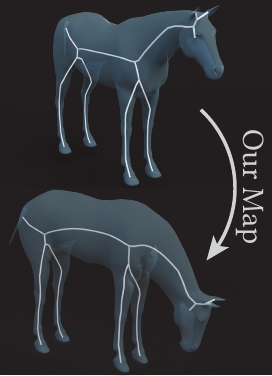}
  \caption{\resub{Internal curve-skeleton transfer.}}
  \label{fig:horse-skeleton}
\end{wrapfigure}
The pushed-forward integral curves closely match the integral curves computed directly on $M_2$, also reflecting the map's degree of metric preservation.

In \resub{another} example, we simulate an internal geometry transfer task. As shown in Fig.~\ref{fig:organs}, we place several objects representing anatomy inside of our source mesh and push these forward to our target. Despite rotation of the head and movement of the arm, structure is largely maintained. For the meshes used in this example we credit~\cite{urban2015heart,yeg3d2015brain,jim2018eye,prevue2013spine,webnode2017bone}.

\resub{In a final example, we transfer a curve-skeleton of a horse mesh to our target (Fig.~\ref{fig:horse-skeleton}). The source skeleton is generated using the approach of~\citet{cao_smi10}. The transferred skeleton captures the deformation of the horse, as evidenced by the curvature of the spine. Previous work has proposed skeleton transfer by finding a rigid transformation between skeletons of two surface meshes~\cite{seylan20193d}. In contrast, our volumetric approach facilitates internal geometry transfer and does not require computing matchings of internal shapes. }

\subsection{Hex mesh transfer}

\begin{figure}
    \includegraphics[width=\linewidth]{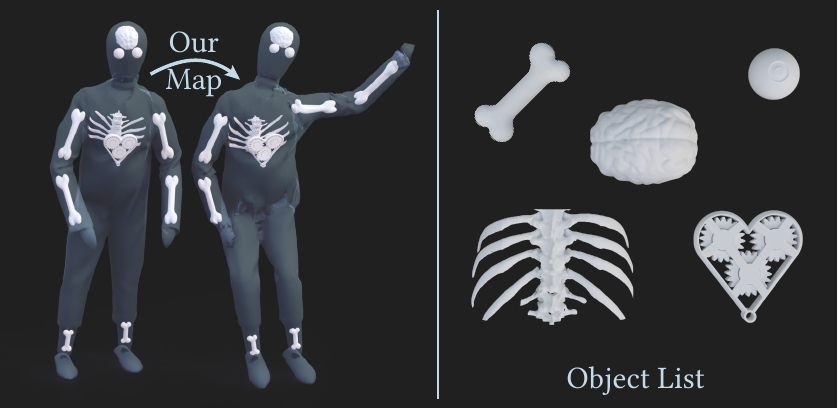}
    \caption{Internal geometry transfer. We place several objects representing human anatomy in the interior of our source mesh and push these forward to the target using our volumetric map.}
    \label{fig:organs}
\end{figure}

Our maps can transport other volumetric structures. Hexahedral meshing remains difficult and often requires extensive human intervention; our maps can  transport expensive-to-compute hex meshes between domains. \figref{hexmesh} transports a hexahedral mesh designed using the method of \citet{li2021interactive} on one domain to a deformed domain. Similar to how we push forward integral curves, we transport a hex mesh by mapping its vertices individually, maintaining the combinatorial structure of the mesh. Due to the low metric distortion of the map, the distortion of most of the hexahedra remains low, as measured by the scaled Jacobian. However, the right foot of the mapped hex mesh has two toes joined together. This artifact is caused by projection to the wrong boundary target, an artifact also encountered by~\citet{li2021interactive}; as their approach has user interaction, they suggest adding landmarks during the optimization to clarify difficult targets. 

\subsection{Volumetric data transfer}

We demonstrate one example of volumetric data transfer using a dataset of placentas extracted from fetal MRI~\cite{abulnaga2021volumetric}. The mapping is done on data from two patients. The first mapped pair contains two scans acquired where the mother is lying in two positions: supine and left lateral. The second contains two scans acquired $\sim\!10$ minutes apart.  \resub{Fig.~\ref{fig:plac} presents the results}. The figure marks one important anatomical landmark, a cotyledon, which is responsible for the exchange of blood from the maternal side to the fetal side~\cite{benirschke1967pathology}. Cotyledons appear as hyperintense circular regions in MRI. We observe close correspondence in the placental geometry. Similar patterns are seen in the mapped texture and the target. \resub{In this application, neither example has a clearly defined source or target shape. The symmetry in our method is advantageous for downstream tasks, such as statistical shape analysis or label propagation, as it prevents bias caused by arbitrarily selecting a source and a target.}
We leave to future work a detailed study of our method's relevance to MRI data.

\begin{figure}
    \includegraphics[width=\linewidth]{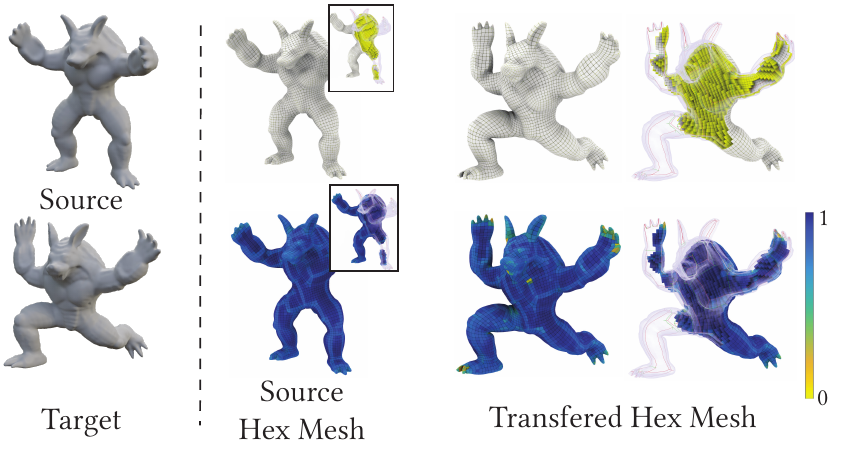}
    \caption{Hex mesh pushed forward from one volume to another using our map.
    We observe low distortion, as measured by the scaled Jacobian overall, but there is some distortion in the mapped right foot.
    Hex meshes are visualized with HexaLab \cite{bracci2019hexalab}, which clamps negative scaled Jacobian values to 0.\vspace{-.1in}}
    \label{fig:hexmesh}
\end{figure}
\section{Discussion}

We successfully map a collection of shapes of diverse geometry and demonstrate that our maps closely match the target boundary with low distortion throughout the volume \resub{and a negligible amount of flipped tetrahedra}. \resubb{Our method is robust to the choice of initialization (Figs.~\ref{fig:optimization-graph},~\ref{fig:landmark}, and~\ref{fig:danielle-comp}) and can produce a dense correspondence even when starting with a low-quality, many-to-one map (Fig.~\ref{fig:optimization-graph} and~\ref{fig:landmark})}. \resub{Compared to the baseline, our free boundary-based approach can recover from poor initialization (Fig.~\ref{fig:danielle-comp} and produce higher quality maps as shown in Table~\ref{tab:map-comparison-baseline-k} and Fig.~\ref{fig:baseline})}. Our examples illustrate scenarios that require a volumetric correspondence, namely internal geometry transfer, hex mesh transfer, and volumetric data transfer. 

Key to the development of our algorithm was the analysis of symmetric distortion energies in \S\ref{s:symm-energy}-\ref{sec:symm-design}. We symmetrized several common distortion energies and found that only the sARAP energy had the desirable properties of favoring isometry, preserving structure, and being nonsingular. We provide a simple way to symmetrize a distortion energy and check if it satisfies these properties. Fig.~\ref{fig:arapanddirichlet} also shows that some choices  of energy can lead to degenerate maps that are sensitive to the parameters used. 
The nonsingularity of the sARAP energy is favorable for computing a map given a degenerate initialization. Since volumetric correspondence has no obvious initializer, this property is key in our target applications, as we can recover from poor initialization.
Future work remains in designing symmetric distortion energies that satisfy more desirable properties.

The connection between the theoretical analysis in \S\ref{s:symm-energy}--\ref{sec:symm-design} to our algorithm design relies on $\psi = \phi^{-1}$. We use soft constraints to encourage a bijection and produce maps with low reversibility energy ($E_R=(1.47\pm1.9) \times 10^{-3}$) and few flipped tetrahedra ($7.7 \pm 9.9$).
In practice, we cannot guarantee $\psi = \phi^{-1}$ as our initialization is non-invertible and the existence of an invertible map is not guaranteed.
However, our experimental results demonstrate the theoretical analysis is relevant, as our computed maps favor isometries ($\det \hat{J} =0.98\pm0.02$) and do not collapse (Fig.~\ref{fig:arapanddirichlet}). It remains an open problem to guarantee $\psi=\phi^{-1}$. 

\begin{figure}[t]
    \centering
    \includegraphics[width=\linewidth]{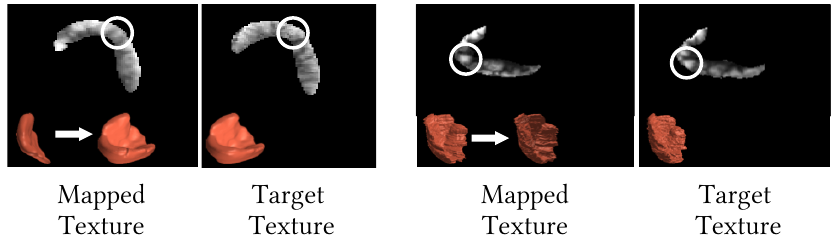}
    \caption{Volumetric data transfer of two fetal MRI volumes visualized as cross-sections of 3D MRI. 
    The figure shows texture transfer between two volumes in a scenario where the mother is lying in the supine and left lateral position \emph{(left)}, and in a scenario where the two volumes are approximately 10 minutes apart \emph{(right)}.
    The circle marks the location of a cotyledon in the target texture.}
    \label{fig:plac}
\end{figure}

\subsection{Limitations}

\begin{figure}[t]
    \includegraphics[width=\linewidth]{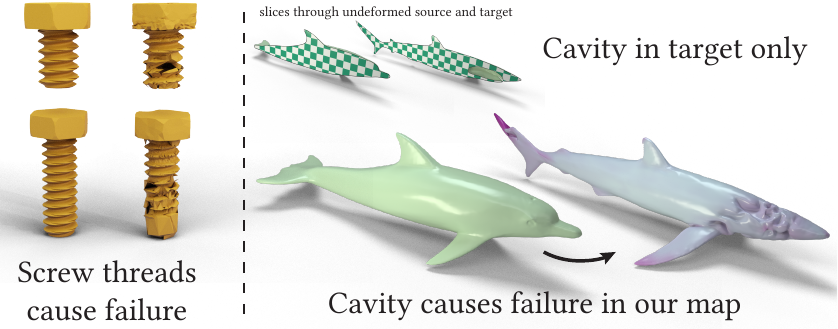}
    \caption{Limitations. We were unable to map between the screw threads, as the map required removing or adding a large amount of material, leading to significant distortion. In our second example, the target shape, a shark, had a large cavity in its interior, while the source, a dolphin, did not.}
    \label{fig:limitations}
\end{figure}

We observed a few failure cases as can be seen in Fig.~\ref{fig:limitations}. First, we encountered shapes where finding a volumetric map was simply infeasible. In the screw threads example, the required map would have to add or remove a large amount of volume, which would lead to substantial distortion. Furthermore, the threads on the boundary differ in number, making it impossible to match sharp features. In the second case, we were unable to map a shark with a cavity in its interior to a dolphin with a solid interior. The cavity is a large hollow area to which a volumetric approach is highly sensitive. \resubb{Furthermore, our method is unable to change topology when mapping between shapes of different genus (Fig. \ref{fig:limitation-topology}) and  we are unable to prescribe topological constraints.} \resub{Another limitation is that our method may not be suitable for partial volume matching, since we normalize input meshes to have volume $1$.} Last, as demonstrated in Fig.~\ref{fig:hexmesh}, our method can join together small features in the boundary (e.g., armadillo toes).
This artifact is caused by an incorrect boundary projection. A potential fix would be to have soft landmark constraints in the optimization.
\begingroup
\setlength{\intextsep}{6pt}%
\setlength{\columnsep}{12pt}%
\begin{wrapfigure}{r}{0.18\textwidth}
  \centering
    \includegraphics[width=0.17\textwidth]{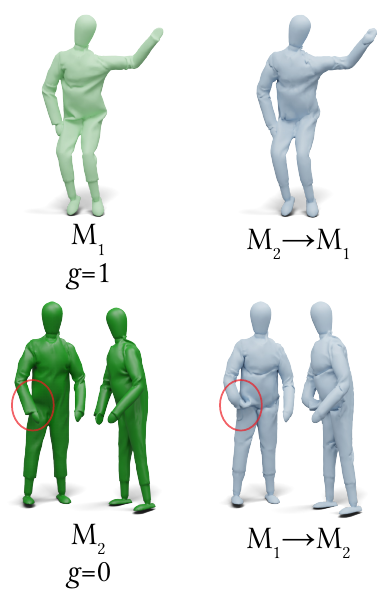}
  \caption{\resubb{Highly distorted region (red circle) when mapping from a genus-1 to a genus-0 shape.}}
  \label{fig:limitation-topology}
\end{wrapfigure}

Finally, our method takes between minutes and hours to compute the correspondences. \resubb{The computational cost is problematic if desiring mapping a collection of shapes, despite our algorithm being advantageous in that we can map between shapes that are far-from-isometries, and we do not require the same connectivity between shapes}. The computational bottleneck is computing the SVD for each tetrahedron many times on the CPU to approximate the gradient of the objective function. A future direction is to improve the convergence time by using a second-order method and to use the GPU for parallelization.

\subsection{Future Work}

An exciting future direction is to develop application-specific volumetric correspondences. We provided a few examples of tasks where volumetric correspondence is useful. Our example of mapping MRI signals demonstrated that while matching geometries can improve correspondence, a method that incorporates both the geometry and signal intensities is needed. \resubb{One framework could be to combine our vertex-based approach with functional maps. }

We envision this work to be a starting point for dense volumetric correspondence applicable to a broad set of shapes. The nascent area of volumetric correspondence is largely unexplored, and our theoretical discussion suggests many intriguing mathematical questions and algorithmic design challenges.

\begin{acks}
This work is funded in part by NIH NIBIB NAC P41EB015902, NIH NICHD R01HD100009, Wistron Corporation, Army Research Office grants W911NF2010168 and W911NF2110293, of Air Force Office of Scientific Research award FA9550-19-1-031, of National Science Foundation grants IIS-1838071 and CHS-1955697, from the CSAIL Systems that Learn program, from the MIT--IBM Watson AI Laboratory, from the Toyota--CSAIL Joint Research Center, from a gift from Adobe Systems, the Swiss National Science Foundation’s Early Postdoc.Mobility fellowship, the NSF Graduate Research Fellowship, the NSERC Postgraduate Scholarship -- Doctoral, and the MathWorks Fellowship.

The authors would also like to thank David Palmer for insightful discussions and help with the frame field and hex mesh transfer experiments. We are grateful to Lingxiao Li for his help in extending the CUDA code for tetrahedron projection in $\R^6$. We would like to thank Shahar Kovalsky for his help in setting up the baseline code and Melinda Chen for her help in generating the curve-skeleton. Finally, we are grateful to Yu Wang, Zo\"{e} Marschner, Artem Lukoianov, and Ishita Goluguri for their help in proof reading.
\end{acks}


\bibliographystyle{ACM-Reference-Format}
\bibliography{main}

\appendix
\label{appendix:table}
\begin{table*}[b]
    \scriptsize
    \centering
        \caption{\resub{Quantitative results on all mesh pairs in our dataset. Our maps closely match the target boundaries while producing low distortion and few tetrahedron inversions. Here, $\ell$ denotes the number of landmarks. Time is measured in minutes.}}
            \label{tab:all-stats-rhm-init}
   
     \begin{tabular}{|l|l|r|r|r|r|r|r|r|r|r|r|r|r|r|r|r|r|}
    \hline
\multicolumn{2}{|l|}{Names}  & $\ell$ & Time & \multicolumn{2}{c|}{$|\mathcal{T}_i|$} & \multicolumn{2}{c|}{$E_{R}  (\times 10^{-3})$} & \multicolumn{2}{c|}{$E_{ARAP}  (\times 10^{-3})$} &\multicolumn{2}{c|}{$n_{inv}$} & \multicolumn{2}{c|}{${\hat d}_{\max}  (\times 10^{-2})$} & \multicolumn{2}{c|}{${\hat d}_{\mathrm{avg}} (\times 10^{-2})$} & \multicolumn{2}{c|}{$|\det \hat J|$} \\ \hline
                scan\_011 & scan\_019 & 23 & 29.23 & 36420 & 43527 & 2.23 & 1.52 & 55.89 & 64.91 & 3 & 11 & 1.47 & 7.82 & 0.12 & 0.24 & 0.98±0.05 & 0.974±0.0545 \\ \hline
        scan\_011 & scan\_030 & 19 & 11.39 & 36420 & 37588 & 0.12 & 0.11 & 43.18 & 43.22 & 2 & 0 & 1.83 & 2.52 & 0.10 & 0.12 & 0.981±0.0372 & 0.984±0.0275 \\ \hline
        scan\_019 & scan\_039 & 16 & 17.57 & 43527 & 50713 & 0.30 & 0.32 & 55.87 & 58.49 & 0 & 2 & 1.57 & 2 & 0.12 & 0.14 & 0.976±0.0503 & 0.979±0.0456 \\ \hline
        airplane1 & airplane2 & 7 & 28 & 24894 & 30700 & 1.29 & 3.02 & 257.28 & 174.17 & 8 & 12 & 2.33 & 2.34 & 0.10 & 0.12 & 0.968±0.0773 & 0.954±0.0979 \\ \hline
        armadillo & deformed armadillo & 21 & 75.89 & 81114 & 113794 & 0.11 & 0.11 & 25.40 & 28.22 & 2 & 3 & 1.15 & 1.39 & 0.05 & 0.06 & 0.99±0.0244 & 0.988±0.0259 \\ \hline
        244.1.ele & 248.1.ele & 12 & 19.55 & 30361 & 90810 & 0.17 & 0.18 & 45.16 & 44.36 & 0 & 7 & 1.23 & 1.15 & 0.06 & 0.08 & 0.987±0.0278 & 0.984±0.0318 \\ \hline
        cat0 & cat1 & 17 & 9.86 & 17867 & 22988 & 0.51 & 0.45 & 40.33 & 51.52 & 0 & 1 & 1.98 & 1.73 & 0.08 & 0.08 & 0.984±0.0479 & 0.981±0.0467 \\ \hline
        cat4 & cat5 & 18 & 14.05 & 25985 & 22710 & 0.92 & 1.02 & 52.71 & 56.03 & 5 & 5 & 3.92 & 2.89 & 0.13 & 0.11 & 0.98±0.0498 & 0.98±0.0496 \\ \hline
        centaur0 & centaur1 & 37 & 13.70 & 30357 & 26954 & 0.42 & 0.38 & 26.06 & 29.56 & 1 & 3 & 1.04 & 1.42 & 0.05 & 0.07 & 0.993±0.0212 & 0.988±0.0299 \\ \hline
        dancer & dancer2 & 13 & 43.11 & 58535 & 36902 & 7.32 & 4.06 & 268.53 & 287.04 & 41 & 17 & 1.32 & 2.18 & 0.12 & 0.11 & 0.934±0.131 & 0.951±0.0909 \\ \hline
        dog4 & dog5 & 27 & 17.04 & 31469 & 30160 & 1.58 & 2 & 61.45 & 53.25 & 7 & 3 & 2.29 & 2.27 & 0.09 & 0.10 & 0.979±0.0527 & 0.98±0.0538 \\ \hline
        dog6 & dog7 & 24 & 33.92 & 26739 & 43771 & 3.87 & 4.34 & 137.22 & 121.09 & 3 & 14 & 7.01 & 2.38 & 0.20 & 0.17 & 0.961±0.0754 & 0.958±0.0863 \\ \hline
        dog7 & dog8 & 25 & 71.15 & 81145 & 85128 & 0.34 & 0.35 & 22.03 & 22.51 & 4 & 9 & 2.86 & 4.78 & 0.08 & 0.12 & 0.992±0.031 & 0.992±0.0286 \\ \hline
        Dolphin & Shark & 9 & 39.75 & 129443 & 55440 & 2.51 & 2.90 & 123.05 & 80.68 & 11 & 13 & 3.96 & 5.32 & 0.14 & 0.16 & 0.981±0.0576 & 0.981±0.0592 \\ \hline
        dragon\_stand & dragonstand2 & 28 & 70.37 & 109823 & 194651 & 0.01 & 0.01 & 0.60 & 0.32 & 0 & 0 & 0.75 & 0.95 & 0.02 & 0.02 & 1±0.002 & 1±0.003 \\ \hline
        fish1 & fish2 & 8 & 73.57 & 64410 & 58215 & 3.07 & 2.07 & 149.62 & 179.76 & 55 & 16 & 2.10 & 2.09 & 0.16 & 0.12 & 0.969±0.0872 & 0.969±0.0763 \\ \hline
        glass1 & glass2 & 13 & 12.83 & 30921 & 13439 & 7.28 & 7.81 & 273.30 & 283.75 & 24 & 0 & 1.85 & 1.33 & 0.22 & 0.10 & 0.918±0.122 & 0.89±0.107 \\ \hline
        gorilla1 & gorilla5 & 26 & 30.52 & 37417 & 59375 & 1.26 & 1.24 & 33.04 & 58.37 & 1 & 4 & 4.05 & 2.76 & 0.09 & 0.10 & 0.988±0.0326 & 0.978±0.0448 \\ \hline
        horse0 & horse5 & 16 & 20.55 & 31507 & 34978 & 0.23 & 0.21 & 31.10 & 35.11 & 0 & 1 & 2.16 & 2.37 & 0.05 & 0.06 & 0.989±0.0414 & 0.99±0.027 \\ \hline
        Cow\_t & Horse\_t & 21 & 20.09 & 31694 & 32515 & 0.58 & 1.16 & 117.74 & 129.02 & 17 & 18 & 2.25 & 3.43 & 0.13 & 0.23 & 0.978±0.0515 & 0.969±0.0641 \\ \hline
        human1 & human2 & 21 & 42.27 & 56550 & 82581 & 0.57 & 0.87 & 58.34 & 60.89 & 0 & 33 & 1.67 & 2.09 & 0.06 & 0.10 & 0.988±0.0267 & 0.985±0.0428 \\ \hline
        michael0 & michael7 & 20 & 19.53 & 19445 & 30014 & 0.40 & 0.35 & 21.62 & 27.28 & 1 & 4 & 1.75 & 1.46 & 0.06 & 0.07 & 0.992±0.0221 & 0.991±0.0269 \\ \hline
        seahorse2 & seahorse4 & 22 & 8.24 & 13720 & 15667 & 0.09 & 0.11 & 16.89 & 17.14 & 1 & 1 & 1.10 & 1.91 & 0.04 & 0.04 & 0.993±0.0227 & 0.993±0.0213 \\ \hline
        toy1 & toy2 & 12 & 31.54 & 75236 & 62880 & 0.43 & 0.47 & 47.36 & 53.13 & 5 & 3 & 4.74 & 3.66 & 0.08 & 0.07 & 0.987±0.0345 & 0.979±0.05 \\ \hline
            \end{tabular}
\end{table*}

\end{document}